\documentclass[11pt]{article}
\usepackage[letterpaper,margin=1in]{geometry}

\usepackage[utf8]{inputenc} 
\usepackage[T1]{fontenc}    
\usepackage{url}            
\usepackage{booktabs}       
\usepackage{amsfonts}       
\usepackage{nicefrac}       
\usepackage{microtype}      
\usepackage{xcolor}         

\usepackage{nicefrac,bm,bbm}
\usepackage{amsmath,amsthm,color,colortbl,amssymb}

\usepackage[hidelinks]{hyperref}

\usepackage{enumitem}
\usepackage{caption}
\usepackage{natbib}

\usepackage{wrapfig}
\usepackage{comment}

\usepackage{graphicx}
\usepackage{multirow}
\usepackage{xspace}
\usepackage{mathtools}
\usepackage{euscript}
\usepackage{thm-restate}
\usepackage{enumitem}
\usepackage{tikz}
\usepackage{subcaption}
\usepackage{tablefootnote,tabularx}
\usepackage[T1]{fontenc}    

\usepackage{pifont}

\usepackage{euscript}
\usepackage{mleftright}

\usepackage{algorithm}
\usepackage[noend]{algpseudocode}

\usepackage[capitalize]{cleveref}

\crefalias{AlgoLine}{line}%
\crefname{algocf}{Algorithm}{Algorithms}

\usepackage{cancel}
\usepackage{newpxtext}
\usepackage{newpxmath}

\usepackage{amsthm,xspace,xcolor}

\newcommand{\rev}{\textsf{rev}}
\newcommand{\revt}{\textsf{rev}_t}

\newcommand{\augment}{\textsc{Augment-the-Best-Mechanism}\xspace}

\newcommand{\E}[1]{\mathbb{E}\left[#1\right]}
\newcommand{\Ei}[1]{\mathbb{E}^i\left[#1\right]}
\newcommand{\Eo}[1]{\mathbb{E}^0\left[#1\right]}

\newcommand{\Ea}[1]{\mathbb{E}^{\alpha}\left[#1\right]}
\renewcommand{\P}[1]{\mathbb{P}\left(#1\right)}
\newcommand{\Pb}{\mathbb{P}}

\newcommand{\Po}[1]{\mathbb{P}^0\left(#1\right)}

\renewcommand{\Pi}[1]{\mathbb{P}^i\left(#1\right)}

\newcommand{\kl}[2]{\mathcal{D}_{\mathrm{KL}}\left(#1,#2\right)}
\newcommand{\tv}{\mathrm{TV}}

\newcommand{\eps}{\varepsilon}
\newcommand{\supp}{\mathrm{supp}}
\newcommand{\Eeps}{E_\eps}
\newcommand{\Veps}{V_\eps}
\newcommand{\Geps}{G_\eps}

\DeclareMathOperator*{\argmin}{arg\,min}

\newcommand{\A}{\mathcal{A}}

\newcommand{\cE}{\mathcal{E}}
\newcommand{\F}{\mathcal{F}}
\newcommand{\M}{\mathcal{M}}
\newcommand{\Q}{\mathcal{Q}}
\newcommand{\cS}{\mathcal{S}}

\newcommand{\pL}{\textsc{PathLearning}\xspace}

\newcommand{\ind}[1]{\mathbbm{1}_{\{{#1\}}}}

\newtheorem{lemma}{Lemma}
\newtheorem{claim}{Claim}
\newtheorem{example}{Example}

\newtheorem{corollary}{Corollary}
\newtheorem{proposition}{Proposition}

\newtheorem{theorem}{Theorem}
\newtheorem{definition}{Definition}

\title{Selling Joint Ads: A Regret Minimization Perspective}

\author{
    Gagan Aggarwal$^\star$ \quad
    Ashwinkumar Badanidiyuru$^\ddagger$ \footnote{This work was done while Ashwinkumar was at Google.}  \quad
    Paul D\"{u}tting$^\star$ \quad
    Federico Fusco$^\#$\footnote{Part of this work was done while Federico was intern at Google under the supervision of Paul D\"{u}tting.}  \vspace{6mm}\\
    $^\star$\ Google Research\\
    $^\ddagger$\ Uber\\
    $^\#$\ Sapienza University of Rome\vspace{2mm}\\
    {\textcolor{black}{\small\texttt{\{gagana,duetting\}@google.com}, \texttt{ashwinkumarbv@uber.com,}}}
    {\textcolor{black}{\small\texttt{fuscof@diag.uniroma1.it}}}
}

\date{}
\begin{document}

\maketitle
\thispagestyle{empty}

\begin{abstract}
Motivated by online retail, we consider the problem of selling one item (e.g., an ad slot) to two non-excludable buyers (say, a merchant and a brand). 
This problem captures, for example,  situations where a merchant and a brand cooperatively bid in an auction to advertise a product, and both benefit from the ad being shown. A mechanism collects bids from the two and decides whether to allocate and which payments the two parties should make. This gives rise to intricate incentive compatibility constraints, e.g., on how to split payments between the two parties. We approach the problem of finding a revenue-maximizing incentive-compatible mechanism from an online learning perspective; this poses significant technical challenges. First, the action space (the class of all possible mechanisms) is huge; second, the function that maps mechanisms to revenue is highly irregular, ruling out standard discretization-based approaches.

In the stochastic setting, where agents' valuations are drawn according to some fixed but unknown distribution, we design an efficient learning algorithm achieving a regret bound of $O(T^{\nicefrac 34})$. Our approach is based on an adaptive discretization scheme of the space of mechanisms, as any non-adaptive discretization fails to achieve sublinear regret. 
In the adversarial setting, when the valuations are arbitrarily generated upfront, we exploit the non-Lipschitzness of the problem to prove a strong negative result, namely that no learning algorithm can achieve more than half of the revenue of the best fixed mechanism in hindsight.
We then consider the $\sigma$-smooth adversary, which randomly generates the valuations from smooth distributions but, unlike in the stochastic case, can do so in a non-stationary way. In this setting, we construct an efficient learning algorithm that achieves a regret bound of $O(T^{\nicefrac 23})$ and builds on a succinct encoding of exponentially many experts.
Finally, we prove that no learning algorithm can achieve less than $\Omega(\sqrt T)$ regret, in both the stochastic and the smooth setting, thus narrowing the range where the minimax regret rates for these two problems lie.
    
\end{abstract}

\clearpage

\tableofcontents

\clearpage

\pagenumbering{arabic}

\section[Introduction]{Introduction 
{}}
\label{sec:introduction}

    Consider an online retail website such as Amazon or Alibaba. These websites sell goods directly but also serve as middlemen between buyers and sellers. Moreover, they typically feature ads by which a certain offer can appear more prominently on a search results page. A crucial feature of these ads is that they may simultaneously benefit a brand and a merchant (say Nike and Footlocker). This is precisely what happens in Facebook’s ``Collaborative Ads'' program \citep{meta2024}, where brands and merchants cooperate to purchase joint ads. 
    
    Situations like this, where a brand and a merchant cooperatively bid in an auction, are a prime motivation for designing selling mechanisms for a single, non-excludable 
    good \citep{Guth86}. The mechanism --- whose goal is revenue maximization --- collects bids from the merchant and the brand and decides whether it wants to allocate the good (show the ad). 
    In the cases where it allocates, the mechanism decides on a payment and how it will be split between the two parties. We call the problem of designing such revenue-maximizing mechanism 
    the \emph{Joint Ads Problem}. 
    Non-excludable mechanism design problems like this give rise to intricate incentive constraints: 
    Unless we set up the mechanism in an incentive-compatible way, the merchant (or the brand) may try to shade their bid, thereby transferring part of the costs to the other party. An incentive-compatible mechanism prevents such gaming opportunities and ensures that it is in the agents' best interest to report their valuations truthfully. While the ``one-shot'' task of selling a single, non-excludable, public good is well understood \citep{Guth86,csapo2013optimal}, it requires strong informational assumptions, e.g., the agents' valuations need to be drawn (independently) from regular or finitely supported distributions which are perfectly known. For practical purposes, a more data-driven approach is desirable. 

    Following \citet{KleinbergL03}, 
    we approach the design of an incentive-compatible revenue-maximizing mechanism for the Joint Ads Problem from an online learning perspective. 
    Formally, we study the Repeated Joint Ads problem, where a mechanism designer interacts at each time step $t = 1, \ldots, T$, with a new pair of agents (e.g., merchant and brand) that want to purchase an item (e.g., an ad slot). Each of these two agents $i \in \{ 1,2\}$ is characterized by a private valuation $v_i^t$ representing its willingness to pay. The mechanism designer (that we often refer to as learner) posts a dominant-strategy-incentive-compatible (DSIC) and individually-rational (IR) mechanism $M^t$. 
    The agents' valuations and the mechanism induce a revenue $\rev_t(M^t)$ for the learner, given by the sum of the two agents' payments. 
    The mechanism designer's goal is to maximize revenue, and the performance of a learning algorithm is measured in terms of its regret: the difference between the revenue achieved and that of the best fixed DSIC and IR mechanism. 
    
    The best fixed mechanism is an ambitious benchmark, which differentiates us, for instance, from the line of work on second price auctions with reserve price(s) \citep[e.g.,][]{Cesa-BianchiGM15,RoughgardenW19}, where the regret is measured with respect to the best fixed reserve price(s). 
    We investigate three possible models generating the agents valuations: (i) the adversarial setting, where valuations are generated up-front in an arbitrary way; (ii) the stochastic setting, where valuations are drawn i.i.d. from a fixed but unknown distribution; and 
    (iii) the $\sigma$-smooth setting \citep{Haghtalab20,HaghtalabRS21}, where the valuations are drawn from a sequence of distributions that are not too concentrated (see \Cref{sec:preliminaries} for formal definitions). 
    Note that, in the stochastic i.i.d. case, the benchmark coincides with the best Bayesian mechanism, but unlike in the ``one-shot'' Bayesian mechanism, the input distribution is not known in advance and has to be learned on the fly. 
    Settings (ii) and (iii) are both considered weaker adversaries than (i), as they consider valuations that are generated randomly, but are incomparable to each other. In (ii), valuations are drawn from a fixed distribution, which can be ``spiky.'' In (iii), valuations are non-stationary but smooth.

\subsection[Our Results]{Our Results {}}
\label{sec:results}

    We provide an overview of our results for the Repeated Joint Ads problem, 
    and compare them with the closest results in the literature. While our main motivation lies in its practical importance, it is also arguably the most fundamental problem for initiating the study of non-excludable mechanism design from a no-regret learning perspective. 
    \begin{itemize}
        \item In the stochastic setting, we provide an efficient learning algorithm, \augment, that exhibits a regret of $\tilde O(T^{\nicefrac 34})$ (\Cref{thm:stochastic upper}), where $\tilde O$ hides poly-logarithmic terms. This upper bound is complemented by an $\Omega(\sqrt{T})$ lower bound (\Cref{cor:lower_stochastic}).
        \item In the adversarial setting, we prove that no learning algorithm can achieve sublinear regret with respect to the best fixed mechanism in hindsight. We actually present a stronger result: it is impossible to even approach {\em half} of the revenue of the best mechanism (\Cref{thm:lb-adversarial}).
        \item In the $\sigma$-smooth setting, we design an efficient learning algorithm, \pL, which achieves an $O(T^{2/3})$ regret  (\Cref{thm:hedge}). Furthermore, we prove that any algorithm suffers regret $\Omega(\sqrt{T})$ in such setting (\Cref{thm:lower_smooth}). 
    \end{itemize}

    Together, our results paint a sharp separation between the two settings with randomly generated input
    and the adversarial case. 
    As we detail below, this divide is caused by the fact that the random nature of the valuations, which is slightly different in the two models, 
    enables two different ways of dealing with the complexities of the model (size of the action space, and irregularity 
    of the objective function), which proves insurmountable when the input is adversarial.

    Our work can be seen as generalizing the model of \citet{KleinbergL03} (also previously studied by \citet{BlumKRW03}) from one buyer to two buyers. 
    The first difference between the single-buyer and multi-buyer cases is that posted pricing is optimal in the former. The natural type of feedback in the single-buyer case is thus \emph{censored feedback} (i.e., whether the buyer purchased or not); while in our model, optimal mechanisms are complex, and the valuations are necessary to implement payments (see \Cref{sec:DSIC}). A second feature differentiating our model from theirs is that in the single-buyer case, 
    the objective function, i.e., the function that maps prices to revenue, is left-Lipschitz. These differences lead to provably different regret regimes. In particular, in the single-buyer case, in the full feedback model, 
    their results show no separation between the stochastic and the adversarial setting (the minimax regret in both cases is $\Theta(\sqrt{T})$). 
    Our positive result in the $\sigma$-smooth setting, as opposed to the negative one in the adversarial case, further emphasizes the relevance of the smoothness paradigm as a tool to overcome adversarial lower bounds while retaining its non-stationary behavior. This phenomenon has already been observed for, e.g., general learning problems \citep{Haghtalab20,HaghtalabRS21}, bilateral trade \citep{CesaCCFL23smoothed}, and first price auctions \citep{CesaCCFL23}. 
    Closely related, \citet{DurvasulaHZ23} investigate smoothness in revenue-maximizing auctions. However, our results are not subsumed by theirs as their work considers the incomparable setting where, at most one of $n$ buyers gets the item. 

\subsection[Challenges  and Techniques]{Challenges  and Techniques {}}\label{sec:challenges}
\begin{figure*}[ht!]
	\captionsetup[subfigure]{aboveskip=0.5pt}
	\centering
	\begin{subfigure}{.33\textwidth}
		\centering
 		\scalebox{0.33}{\centering\includegraphics{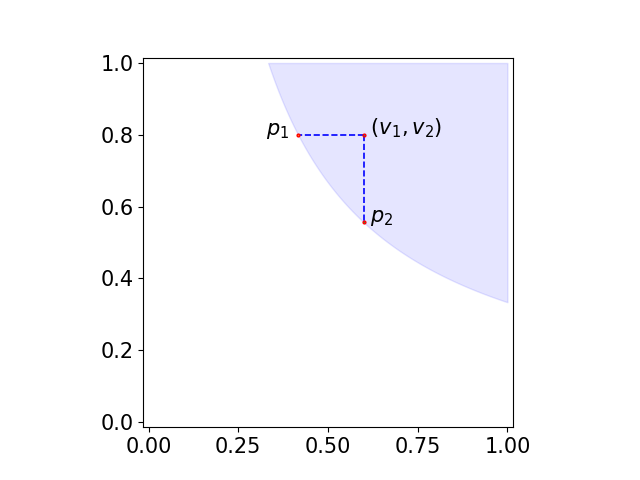}}
		\caption{\footnotesize Concave allocation region}
		\label{fig:concave}
	\end{subfigure}\hspace{0.1pt}%
	\begin{subfigure}{.33\textwidth}
		\centering
		\scalebox{0.33}{\includegraphics{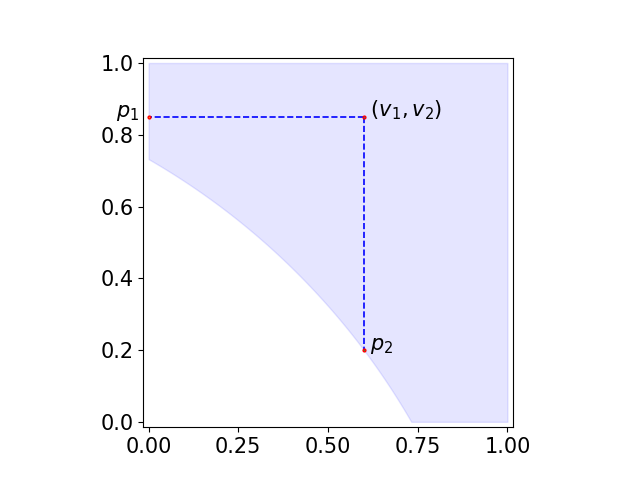}}
		\caption{\footnotesize Convex allocation region}
		\label{fig:convex}
	\end{subfigure}\hspace{0.1pt}%
	\begin{subfigure}{.33\textwidth}
		\centering 		\scalebox{0.33}{\includegraphics{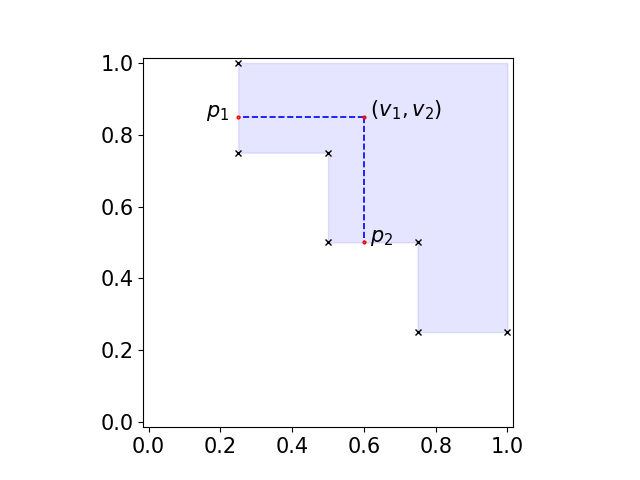}}
		\caption{\footnotesize Complex allocation region}
		\label{fig:complex}
	\end{subfigure}
	\caption{\small The allocation regions of optimal mechanisms corresponding to different distributions. The payment functions can be visualized as the valuations projection on the allocation region's boundary. For instance, in \Cref{fig:convex}, the first agent (x-axis) pays $p_1=0$, while the second agent (y-axis) pays $p_2 = 0.2$. The valuations are drawn i.i.d. from the distribution with cumulative density function $F(x) = 2(x+1)^{-2}$ in \Cref{fig:concave} and $F(x) = x^2$ in \Cref{fig:convex}. 
 In \Cref{fig:complex}, the distribution has support on the black points.} 
	\label{fig:regions}
\end{figure*}
    
    The main technical challenge in the Repeated Joint Ads Problem is the action space: the set $\M$ of all DSIC and IR mechanisms is not only huge but also complex. Furthermore, the function mapping mechanisms to revenue become highly non-regular when the valuation distribution has point masses. A mechanism is characterized by two functions- an allocation and a payment rule- mapping bids to allocation and payments. 
    Standard mechanism design arguments 
    \citep{myerson81} enable a simple geometric interpretation of DSIC and IR mechanisms for our problem. The allocation rule of any mechanism corresponds to a subdivision of the unit square $[0,1]^2$ into two regions: one in which the good is allocated and one where it's not. For DSIC and IR, the two regions must be ``north-east'' monotone so that increasing either valuation can only lead to a better allocation. The payments are then given by the ``south-west'' projection of the valuation pair on the boundary of the allocation region. See \Cref{fig:regions} for a visualization. This simplifies the problem from learning two functions to a single curve. However, the curve defining the optimal mechanism can be complex, even in simple settings. For instance, even in the (independent) Bayesian setting, the boundary of the optimal mechanism can be convex or concave (see \Cref{fig:concave,fig:convex}).
    In general, for a fixed set of (possibly adversarially given) valuations pairs, the optimal allocation boundary may be neither convex nor concave (see  \Cref{fig:complex}). This notion of complexity is further reinforced by the fact that standard notions of dimensions in learning theory, e.g., VC-dimension and pseudo-dimension, are unbounded for this class (see also \Cref{app:complexity}). This contrasts to what happens, e.g., in \citet{KleinbergL03}, where the learning goal is {\em one single number}: the best fixed price. 

    \paragraph{A Hard Instance} The standard approach to solving an online learning problem on a large action space $A$ is to find a small subset $A' \subseteq A$ with two properties (i) it is possible to achieve sublinear regret with respect the best action in $A'$ and (ii) the performance of the best action in $A'$ is close to that of the best action in $A$. Typically, this is done by exploiting some regularity of the learning problem, such as convexity or Lipschitzness of the objective function \citep{KleinbergSU19,Hazan16}.  To provide an intuition of the difficulty posed by the Repeated Joint Ads problem, we introduce the following (hard) example based on the equal revenue distribution.
    \begin{example}[Equal-revenue]
    \label{ex:counter-example}
    For any $\delta \in (0,1)$ and integer $n$, consider the random variable $V$:
    \[
        \P{V = (\delta (1-\tfrac 1{2^i}), \tfrac 1{2^i})} = \begin{cases}
            2^{1-n} &\text{ for } i = 1 \\
            2^{i - n- 1} &\text{ for } i = 2, \dots, n
        \end{cases} 
    \]
    The random variable is supported on the $(0,1)$-$(\delta,0)$ segment (red dashed line in \Cref{fig:hard-instance}), so the optimal mechanism $M^{\star}$ is clearly the one that allocates if and only if the valuations fall either on the segment itself or its right. The expected revenue of $M^{\star}$ is the sum of the payments of the two agents, which can be computed via the Myerson payment scheme (see \Cref{sec:DSIC} for details), for an expected revenue of $
        \E{\rev(M^*)} \ge \sum_{i=1}^n \tfrac 1{2^i}\P{V = (\delta (1-\tfrac 1{2^i}), \tfrac 1{2^i})}  \in \Theta(\tfrac{n}{2^n}),$ where we ignored the $\delta$ terms as this is a parameter that can be set arbitrarily small.        
    \end{example}
    \begin{figure*}[t!]
	\centering
 \includegraphics[width=0.46\linewidth]{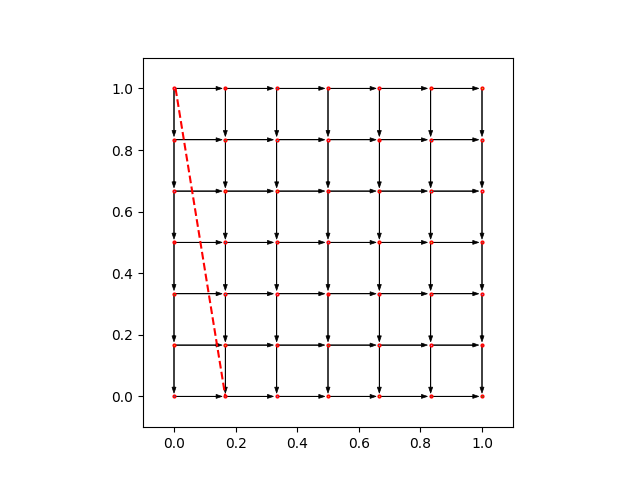}
	\caption{\small The distribution described in \Cref{ex:counter-example} is supported on the red segment (with $\delta = \nicefrac 16$). The Figure also represents the uniform grid for $\eps = \nicefrac 16$.} 
    \label{fig:hard-instance}
\end{figure*}
    We show how the instance in the example fools natural attempts to discretize the mechanisms space: the best mechanism in each one of these families has a multiplicative $\Theta(n)$ gap with respect to the expected revenue extracted by the best mechanism, thus invalidating the hope to achieve sublinear regret (even with respect to a constant fraction of the optimal revenue).
    \begin{itemize}
        \item \textbf{Posted prices.} A first parametric class of mechanisms is provided by posted prices, characterized by rectangular allocation regions. A mechanism in such class offers two prices, $p_1$ and $p_2$, and allocates if both agents accept. This is the immediate generalization of what happens for one agent. However, one can verify that the expected revenue of the best-fixed prices mechanism is $2^{-n}$ (ignoring the first agent who has an arbitrarily small valuation). 
        \item \textbf{Mechanisms on a grid.} A more elaborate approach (similar to what is done, e.g., in \citet{DurvasulaHZ23}) is to consider all the mechanisms whose allocation region can be described by the union of the tiles of a uniform grid (we formalize the class of grid mechanisms in \Cref{sec:stochastic_upper}; see \Cref{fig:grid1}). Surprisingly, this rich class is not more expressive than posted price mechanisms. It is always possible to set $\delta$ smaller or equal to the step size of the grid so that the whole distribution is supported on the first column (see \Cref{fig:hard-instance}).
    \end{itemize}

    \paragraph{The Stochastic Algorithm: An Adaptive Grid.} The hard instance in \Cref{ex:counter-example} hints that we need an adaptive discretization. In particular, we show how to {\em augment} the class of mechanisms associated with uniform grids using the observed samples to approximate the revenue of the best mechanism while maintaining that the class is learnable. Our approach consists of adaptively refining the uniform grid tiles crossed by the boundary of the allocation region of the mechanism (see \Cref{fig:augmented} for a visualization).
    Both the learnability and the approximation properties of these class of ``augmented grid'' mechanisms build on a 
    decomposition of the revenue function, which associates a weight to each edge of the grid and reduces the problem to finding a longest path. Besides achieving sublinear regret, our algorithm \augment can be efficiently implemented. We mention that also the analysis of \citet{csapo2013optimal} features a reduction to a graph problem; although somehow related, there the optimization goal is maximal closure, the decomposition is node-based, and the reduction critically needs finite support of the distribution. 

    \paragraph{The Smooth Algorithm: From Paths to Edges.} As observed, there is a natural class of mechanisms described by the $(0,1)$ to $(1,0)$ paths that are ``down-right'' in the uniform grid $G_{\eps}$, where $\eps$ is the step-size. A crucial observation is that assuming smoothness, the discretization error over $T$ rounds, i.e., the gap between the best mechanism and the best mechanism supported on $G_{\eps}$ is of order $O(\eps T)$ (\Cref{lem:discretization_smooth})\footnote{Note, the instance described in \Cref{ex:counter-example} is not smooth as it is supported on a line.}.  This can be shown by considering the best mechanism, observing through which tiles it passes, and considering the ``inner hull'' 
    mechanism on the grid. This discretization is not enough to get an efficient algorithm with a good regret bound: crucially, in fact, the cardinality of the paths in $G_{\eps}$ is still exponential in $\nicefrac{1}{\eps}$. We observe that standard approaches to deal with large classes of experts either are not implementable (e.g., the Hedge algorithm \citep{AroraHK12}, needs to maintain a weight for each one of the paths) or give bounds worse than our $O(T^{\nicefrac 23})$ (e.g., ``follow the leader'' approaches exhibit way worse bounds, see for instance Theorem 1.1. of \citet{KalaiV05} ). We overcome this problem by designing a careful sampling scheme that allows us to implement the Hedge algorithm efficiently; in particular, with an approach similar in spirit to \citet{TakimotoW03}, we maintain weights on the {\em edges} of $\Geps$ instead of its paths. 

    \paragraph{Lower Bounds.} Proving lower bounds in a structured and continuous environment is challenging, as it entails constructing hard instances while respecting the structure of the problem, i.e., points are valuations mapped to revenue via a specific function. In the lower bound for the stochastic and smooth adversary, we embed the standard hard instance for prediction with expert \citep{nicolo06} into a smooth distribution. For the adversarial lower bound, we use a {\em needle in a haystack} technique to construct an instance characterized by an optimal mechanism --the needle-- that always makes the trade happen and extracts maximum profit from it, while no learning algorithm can extract more than half of the revenue at each iteration.

\subsection{Further Related Work {}}
\label{sec:related}

    
    \paragraph{(Online) Learning of Economic Problems} Following the seminal work of \citet{KleinbergL03}, there is a flourishing line of works that studies economic problems from an online learning perspective. For instance, the following problems have been studied: auctions from both the mechanism designer \citep[e.g.,][]{Cesa-BianchiGM15,RoughgardenW19} and the bidders \citep[e.g.,][]{FengPS18,DaskalakisS22,CesaBianchiCCFS24} perspective, bilateral trade \citep{AzarFF22,CesaCCFL23,CesaCCFL23smoothed,BernasconiCCF24,Babaioff24}, contracts \citep{HoSV14,ZhuBYWYJ22,DuettinGSW23},  prophet inequalities \citep{GatmiryKSW24}, brokerage \citep{BolicCC24}, and Bayesian persuasion \citep{CastiglioniCMG23}. With the exemption of \citet{DaskalakisS22} and \citet{HoSV14,ZhuBYWYJ22}, the positive results in this line of work mostly pertain to cases where the object of learning is not complex. For example, \citet{Cesa-BianchiGM15} consider the problem of learning reserve prices (rather than optimal auctions), and the positive results in \citet{ZhuBYWYJ22,DuettinGSW23} are for linear contracts which are defined through a single parameter. The results in \citet{DaskalakisS22} and \citet{HoSV14,ZhuBYWYJ22} for more complex objects are mostly impossibility results. 
    
    A closely related approach to learning in economic design problems considers the \emph{offline} learning problem in the spirit of PAC-learning \citep[e.g.,][]{ColeR14,MorgensternR15}. While there are known constructions that derive regret bounds for the online problem from sample complexity bounds for the offline problem, we are not aware of any prior work that would imply non-trivial results for our settings. For instance, \citet{ColeR14} consider an excludable setting where bidders’ valuations are independent and make regularity assumptions on the distribution; we consider a non-excludable setting with arbitrary distributions.  \citet{MorgensternR15} provide results only if the pseudo-dimension is finite; in Appendix~\ref{app:complexity} we prove that such measure is unbounded in our problem. \citet{GuoHGZ21} develop an approach that does not rely on any complexity measures of the hypothesis classes, and show general sample complexity bounds for product distributions. Their general bounds imply sample complexity bounds for a range of applications, via appropriate discretizations of the value space. In contrast, a main challenge in our problem is the lack of such discretization (see Example~\ref{ex:counter-example}). Instead, our main result for the stochastic setting is driven by a non-trivial data-dependent discretization of the mechanism space (captured by a carefully designed infinite family of mechanisms, for which \Cref{thm:learnability} provides a sample complexity bound).


    \paragraph{Non-Excludable Mechanism Design.} Questions of non-excludable mechanism design are an active field of research. For example, \citet{HaghpanahKL21} consider the problem of selling to a group where the payment of all group members has to be the same. This model is inspired by situations (e.g., a department purchasing some equipment) where there is only a single payment (to be made by the department), but all agents (members of the department) have different values for the item and need to agree on a purchase decision. \citet{BalseiroMLZ21} tackle instead the design of dynamic mechanisms that exploit the repeated interaction with the same agents (think Cremer-McLean style full-surplus extraction); that work falls into a rather different territory than we do here.

\section[The Learning Model]{The Learning Model }
\label{sec:preliminaries}

    We formally describe the learning protocol for the Repeated Joint Ads problem. We also refer to the pseudocode for further details. At each time step $t = 1, 2,\dots, T$, a new pair of agents arrives, characterized by private valuations $(v^t_1, v^t_2)$ in the $[0,1]^2$ square. Independently, the learner proposes a mechanism $M^t$ to the agents, who then declare bids $(b_1^t,b_2^t)$. A mechanism $M^t$ is composed by a non-excludable allocation rule $x^t: [0,1]^2 \to \{0,1\}$ 
    and two payment functions $p_1^t,p_2^t : [0,1]^2 \to [0,1]$ 
    that map bids to allocation, respectively payment. The revenue of mechanism $M^t$ at time $t$ is denoted with $\revt(M^t)$ and is defined as
    \[        \revt(M^t) = x^t(b^t_1,b^t_2) \cdot [p^t_1(b^t_1,b^t_2) + p^t_2(b^t_1,b^t_2)].
    \]

    \begin{algorithm}[t!]
    \begin{algorithmic}[ht]
        \For{time $t=1,2,\ldots$}
            \State a new pair of agents arrives with (hidden) valuations $(v^t_1,v_2^t) \in [0,1]^2$
            \State the learner declares a mechanism $M^t \in \mathcal{M}$
            \State the agents declare their bids to the learner, allocation and payments are made accordingly
            \State the learner gains $\revt(M^t)$ (relatively to the bids)
        \EndFor
    \end{algorithmic}
    \caption*{\textbf{Learning Protocol for the Repeated Joint Ads Problem}}
    \end{algorithm}

\subsection[Structure of Incentive Compatible Mechanisms]{Structure of Incentive Compatible Mechanisms } 
    \label{sec:DSIC}
    
        The agents behave strategically, meaning that they strive to maximize their own quasi-linear utility $u^t_i(b_1^t,b_2^t) = v^t_i \cdot x^t(b_1^t,b_2^t) - p^t_i(b_1^t,b_2^t)$ for $i \in \{1,2\}$. To avoid unpredictable strategic misreporting from the agents, we require the mechanism $M^t$ proposed by the learner to enforce two properties:
        \begin{itemize}
            \item{\em ex-post dominant strategy incentive compatibility (DSIC)}: regardless of the other's bid, each agent maximizes its utility by being truthful. Formally, for any pair of actual valuations $(v_1^t,v_2^t)$ and possible bids $b_1$ and $b_2$ we have $     u^t_1(v_1^t,b_2) \ge u^t_1(b_1,b_2)$ and $ u^t_2(b_1,v_2^t) \ge u^t_2(b_1,b_2)$.
            \item{\em ex-post individual rationality (IR)}: truthfulness never induces negative utility. Formally, for any actual valuations $(v_1^t,v_2^t)$ and possible bids $b_1$ and $b_2$ we have 
            $u^t_1(v_1^t,b_2) \ge 0$ and $u^t_2(b_1,v_2^t) \ge 0.$
        \end{itemize}

        Our problem falls within the framework of single parameter auctions \citep{myerson81}, as each agent is characterized by a single value representing its value for the good for sale. This tells us that $(i)$ an allocation rule is implementable if and only if it is monotone and $(ii)$ for each such allocation rule there exists a unique payment rule that completes it to a DSIC and IR mechanism.\footnote{Technically, payments are unique only up to constant shifts. We adopt the minimal payments to achieve DSIC and IR.}

        \begin{itemize}
            \item{\em monotone allocation rule:} an allocation rule $x$ is monotone if for every bid $b_1 $ and $b_2$, the two functions $x(z,b_2)$ and $x(b_1,z)$ are non-decreasing in $z$.
        \end{itemize}

        Given a monotone allocation rule $x$, the unique payment rule $p = (p_1,p_2)$ that completes it to a DSIC and IR mechanism can be computed via the critical prices, in particular (adopting the convention that the $\argmin$ of the empty set is $0$) we have:
        \[
            \begin{cases}
                p_1(b_1,b_2) = \argmin\{\pi \in [0,1] \mid x(\pi,b_2) = 1\},\\ p_2(b_1,b_2) = \argmin\{\pi \in [0,1] \mid x(b_1,\pi) = 1\}.
            \end{cases}
        \]
        \paragraph{A Geometric View.}
        The notion of monotone allocation rule and the relative payments have a clean geometric interpretation in our setting: any DSIC and IR mechanism $M$ is identified by its allocation region $A_M = \{(v_1,v_2)\in [0,1]^2 \mid x(v_1,v_2) = 1\}.$
        We say that a point $(w_1,w_2)$ dominates point $(v_1,v_2)$ if $w_1 \ge v_1$ and $w_2 \ge v_2$. To respect monotonicity $A_M$ has the following property: if $(v_1,v_2) \in A_M$, and $(w_1,w_2) \in [0,1]^2$ dominates it, then $(w_1,w_2) \in A_M$. We call a subset $A$ of $[0,1]^2$ monotone if such property is respected. A monotone region is enclosed by the sides of the $[0,1]^2$ square and a monotonically non-increasing curve that starts from either the north or the west side of the square and terminates into either the east or south side of it (see Figure~\ref{fig:regions}).
        The total payment is then given by the sum of the orthogonal projections (horizontal and vertical) of the valuations onto the boundary of $A_M$. 
        All in all, this allows us to forget about the strategic nature of the problem by characterizing mechanisms via ``monotone'' allocation regions (or their boundary) and implementing payments accordingly; for this reason, in the rest of the paper, there is no mention of bids as the agents always report their valuations truthfully. Similarly, we always intend mechanisms to be DSIC and IR. 
        
        \paragraph{Technical convention.} We make the technical assumption that allocation regions are always closed. This is without loss of generality as we are interested in revenue maximizing auctions and taking the closure of the allocation region may only increase the revenue. The assumption allows us, for instance, to use $\min$ instead of $\inf$ in the pricing rule and basically says that valuations belonging to the allocation region's boundary result in an allocation.
    
    \subsection[Regret Minimization]{Regret Minimization }

        The goal of the mechanism designer is to maximize its revenue over the $T$ rounds. In particular, the performance of a learning algorithm is measured in terms of its regret: the difference between the revenue achieved and that of the best fixed mechanism in hindsight. For any (possibly randomized) adversary $\cS$ that generates the sequence of valuations, and learning algorithm $\A$, it is possible to define the regret of $\A$ against $\cS$ as follows:
    \[
        R_T(\A,\cS) =\sup_{M \in \M} \E{\sum_{t = 1}^{T} \revt(M) - \sum_{t = 1}^{T} \revt(M^t)}.    
    \]
    The mechanisms $M_t$ are proposed by the learning algorithm $\A$, while the agents' valuations are generated by the adversary $\cS$. $\M$ denotes the class of all the DSIC and IR mechanisms. 
    The expectation in the above formula is with respect to the internal randomness of $\A$ and (possibly) of $\cS$. The goal of the learner is to design an algorithm $\A$ with small regret over all the possible adversaries: $ R_T(\A) = \sup_{\cS}R_T(\A,\cS)$. We consider three models of adversaries:
    \begin{itemize}
        \item In the (oblivious) adversarial setting, an adversary generates the sequence of valuations $\{(v_1^t,v_2^t)\}$ up-front in an arbitrary way.
        \item In the stochastic setting, the valuations are drawn i.i.d. from an unknown and possibly correlated but fixed distribution. 
        \item In the $\sigma$-smooth setting, an adversary chooses a sequence of $\sigma$-smooth random 
        variables $V^t = (V_1^t,V_2^t)$ from which the actual valuations $v^t = (v_1^t,v_2^t)$ are drawn from. We say that a random variable is $\sigma$-smooth if its distribution is $\sigma$-smooth  as defined below. 
    \end{itemize}
    \begin{definition}[\citet{HaghtalabRS21}]
        Let $X$ be a domain supporting a uniform distribution $\nu$. A measure $\mu$ on $X$ is $\sigma$-smooth, for some $\sigma \in (0,1]$, if for all measurable $A \subseteq X$, we have $\mu(A) \le \frac{\nu(A)}{\sigma}$.
    \end{definition}

\section{The Stochastic Upper Bound}
\label{sec:stochastic_upper}

    In this section, we present an (efficient) learning algorithm, \augment, which exhibits a $\tilde O(T^{3/4})$ upper bound on the regret. First, in \Cref{sec:max_flow}, we introduce the family of orthogonal mechanisms that well approximate the whole class of mechanisms for the Joint Ads Problem. Second, in \Cref{sec:orthogonal}, we present a data-dependent discretization scheme focusing on a learnable subfamily of orthogonal mechanisms that retain their approximating power. Finally, in \Cref{sec:adaptive}, we discuss how to combine these ingredients in our algorithm.

\subsection{From Revenues to Distances on a Graph: The Power of Orthogonal Mechanisms}
\label{sec:max_flow}

    Given the theory of single parameter auctions, to which our problem belongs, any DSIC and IR mechanism $M$ for the Joint Ads Problem is characterized by its allocation region $A_M$. The first natural attempt to discretize such action space is to consider allocation regions that can be seen as a union of tiles in a uniform grid; however, we know that such a class is not rich enough to successfully approximate the best mechanism (see the discussion after \Cref{ex:counter-example}). To overcome this, we introduce a generalization of such a family, the class of orthogonal mechanisms. 

    \begin{definition}[Orthogonal Mechanisms]
        A {DSIC and IR} mechanism $M$ is said to be orthogonal if its allocation region's boundary is composed of a finite union of vertical and horizontal segments. We denote the class of all orthogonal mechanisms with $\M^{\perp}$.
    \end{definition}

    Such mechanisms have a natural interpretation as ``south-east'' paths in a suitable class of graphs. 

    \begin{definition}[Orthogonal Graphs]
        A directed graph $G = (V,E)$ is said to be orthogonal if:
        \begin{itemize}
            \item[(i)] Each vertex $u \in V$ is a distinct point $(u_1,u_2) \in [0,1]^2$.
            \item[(ii)] There is a unique source $s=(s_1,1)\in V$ and a unique sink $t=(1,t_2) \in V$.
            \item[(iii)] Each edge $e = (u,v) \in E$ is either vertically pointing down-wards ($u_1 = v_1$ and $u_2 > v_2$) or horizontally pointing right-wards ($u_1 < v_1$ and $u_2 = v_2$), and the segments representing edges may only intersect in their endpoints.
            \item[(iv)] Each {node in} $V \setminus \{s,t\}$ has at least an incoming and an outgoing edge. 
        \end{itemize}
        We say that a path $\pi$ in $G$ from vertex $w$ to vertex $v$ is {\em complete} if $w$ belongs to the north side of the square ($w_2 =1$) and $v$ belongs to the east side of it ($v_1 = 1$).
    \end{definition}

    Property (iv) guarantees that any maximal path starting from a node on the north side of the square terminates on the east side of it (and is thus complete). Furthermore, any orthogonal graph $G$ is acyclic by property (iii). For any complete path $\pi$ in an orthogonal graph, it is natural to identify the mechanism $M_{\pi}$ whose allocation region falls on the ``top-right'' part of the square with respect to $\pi$. Formally, valuation $(v_1, v_2)$ is in the allocation region of $M_{\pi}$ if there exists a node $u = (u_1,u_2)$ in $\pi$ such that $v_1 \ge u_1$ and $v_2 \ge u_2$. We refer to the \Cref{fig:grid1} for a visualization: the allocation region of the orthogonal mechanism is shaded in blue, while the path representing its boundary is red. 
    The relation between orthogonal graphs and mechanisms is explicit in the following Proposition.

    \begin{restatable}{proposition}{orthogonal}
    \label{prop:orthogonal}
        Let $G = (V,E)$ be any orthogonal graph, and $\pi$ any complete path in $G$, then $M_{\pi} \in \M^{\perp}$. Conversely, for any $M \in \M^{\perp}$, there exists a complete path $\pi$ in some orthogonal graph such that $M = M_\pi$.
    \end{restatable}
    \begin{proof}
        Fix any orthogonal graph $G$ and complete path $\pi$ in it, and consider the associated mechanism $M_\pi$, which is the mechanism associated by the allocation region $A_{\pi}$ of all the points that dominate at least a node in $\pi$. Regarding the shape of $A_{\pi}$, its boundary is given by the union of the edges in $\pi$ (which are either horizontal or vertical by definition of orthogonal graphs) and possibly portions of the north or east side of the $[0,1]^2$ square. We must also formally prove that $A_{\pi}$ is monotone. Let $(v_1,v_2)\in A_{\pi}$ and $(w_1,w_2)$ any point that dominates it, then also $(w_1,w_2)\in A_{\pi}$ because, by definition of $M_{\pi}$ there exists $u = (u_1,u_2)$ in $\pi$ such that $w_1 \ge v_1 \ge u_1$ and $w_2 \ge v_2 \ge u_2$. 

        The converse is easy to prove: fix any orthogonal mechanism $M \in \M^{\perp}$ and consider the part of the boundary of $A_M$ that goes from the north to the east side of the square; by definition, such curve is the union of vertical and horizontal segments. Thus, it constitutes a path that respects the properties of an orthogonal graph.
    \end{proof}

    For this reason, we can safely identify complete paths on orthogonal graphs with the corresponding orthogonal mechanisms and orthogonal graphs $G$ with the family of orthogonal mechanisms that complete paths of $G$ can describe. The relation between these two concepts goes even further: the revenue of any orthogonal mechanism is equal to the sum of carefully chosen weights on the associated path. 
    
    \begin{definition}[Edge weights]
    \label{def:weights}
        Fix any orthogonal graph $G$ and any edge $e$ in it. We define its influence region 
        $A_e$ and intrinsic weight $w_e$ as follows: 
        \[
        \begin{cases}
            A_e = [x,1] \times [u_2,v_2) \text{ and } w_e =x \text{ if $e = (u,v)$ with $u_1 = v_1 = x$}  \\
            A_e = [u_1,v_1) \times [y,1]  \text{ and } w_e =y \text{ if $e = (u,v)$ with $u_2 = v_2 = y$} 
        \end{cases}
        \]
        where we adopt the convention that if $v_2 = 1$ in the first case or $v_1 = 1$ in the second case, then the corresponding interval extreme is closed. Fix now any random variable $V = (V_1,V_2)$ on $[0,1]^2$, we define the weight function $w_V : E \to [0,1]$ as $  w_V(e) = w_e \cdot \P{V \in A_e}.$
    \end{definition} 

    Stated differently, we associate to each edge $e$ its contribution to the revenue {\em if} such edge belongs to the path associated with the orthogonal mechanism. (See \Cref{fig:grid2} for a visualization of the allocation regions $A_e$, in red for a vertical edge and blue for a horizontal one).

    \begin{proposition}
    \label{prop:decomposition}
        Consider any random variable $V$ in $[0,1]^2$ describing the agents' valuations, and any orthogonal graph $G = (V,E)$, then for any complete path $\pi$ in $G$ it holds that
        \[
            \E{\rev(M_{\pi})} = \sum_{e \in \pi} w_V(e) = \sum_{e \in \pi} w_e \cdot \P{V \in A_e}.
        \]
    \end{proposition}
    \begin{proof}
        We start by considering the payment made by the first agent. Agent $1$ pays if and only if the valuation $(v_1,v_2)$ falls into the allocation region, partitioned by all the $A_e$ for $e$ vertical edges in $\pi$. If the valuation falls into $A_e$, then the price paid by agent $1$ is computed according to Myerson's rule and is set precisely to $w_e$. All in all, the revenue extracted by the first agent is exactly the sum over all the vertical edges in $\pi_e$ of $w_e \P{V \in A_e}$. A similar argument can be carried over for agent $2$ and the horizontal edges, thus proving the statement. Note that the edges that belong to the west and south sides of the square carry no weight.
    \end{proof}        

    The previous proposition allows us to ``discretize'' the task of computing the revenue-maximizing mechanism. Given an orthogonal graph $G$, it is possible to calculate its best mechanism by simply computing the longest complete path. We can say something more; if the underlying distribution is finitely supported, then it is possible to create an orthogonal graph (whose node set is composed of the grid generated by the points in the support and thus contains them) that is guaranteed to contain an optimal mechanism on one of its paths, therefore allowing for efficient recovery of the optimal mechanism. 

    \begin{theorem}    \label{thm:best_mechanism}
        Consider any finitely-supported random variable $V$ in $[0,1]^2$ describing the agents' valuations, then there exists an orthogonal mechanism $M^\star \in \M^\perp$ that can be computed efficiently (in $|supp(V)|$) that is revenue maximizing:
        \[
            \E{\rev(M^\star)} \ge \E{\rev(M)}, \, \forall M \in \M.
        \]
    \end{theorem}
     \begin{proof}
        First, we prove that we can restrict our attention to study mechanisms whose allocation region is minimal with respect to the inclusion. Second, we fix any such mechanism $M$ and prove that it is indeed orthogonal. Finally, we explain how to efficiently compute such $M^{\star}$.

        We start with the first step. Fix any mechanism $M$, and consider the auxiliary mechanism $M'$ characterized by the following allocation region:
        \[
            A_{M'} = \{(w_1,w_2) \in [0,1]^2 \mid \exists (v_1,v_2) \in \supp(V) \cap A_M \text{ s.t. } w_1 \ge v_1 \text{ and } w_2 \ge v_2\}. 
        \]
        $A_{M'}$ contains all points that are dominated by the valuations in $supp(V)$ that also belong to $A_M$. 
        \begin{claim}
        \label{cl:minimal}
            $A_{M'}$ is the monotone region that contains $\supp(V) \cap A_M$ which is minimal with respect to inclusion. Therefore $A_{M'} \subseteq A_M$ and $\E{\rev(M')} \ge \E{\rev(M)}$.
        \end{claim}
        \begin{proof}[Proof of \Cref{cl:minimal}]
            For any $(u_1,u_2) \in A_{M'}$ there exists a $(v_1,v_2) \in \supp(V) \cap A_M$ which dominates it. But such $(v_1,v_2)$ belongs to $A_M$ which is monotone, so $(u_1,u_2)$ has to belong to $A_M$ by monotonicity. Minimality follows by the same argument. Since $A_M$ contains $A_M'$, it immediately holds that $\E{\rev(M')} \ge \E{\rev(M)}$, as increasing the size of the allocation region without changing the intersection with the support of $V$ only causes the payment to decrease.    
        \end{proof}
                We move to the second part of the proof: any mechanism $M$ whose allocation region is minimal with respect to inclusion is orthogonal. The intuition is that any ``curve'' part of the boundary is not inclusion-minimal and can be rectified without affecting the points in the support of $V$ covered.
        \begin{claim}
        \label{cl:orthogonal}
            Any mechanism $M$ whose allocation region is minimal with respect to inclusion with respect to a subset of $\supp(V)$ is orthogonal.
        \end{claim}
        \begin{proof}[Proof of \Cref{cl:orthogonal}]
            Let $S_x$ be the set of all the values of the first coordinates of the points in $\supp(V)\cap A_M$ union $\{0,1\}$, and let $S_y$ be the analogous set for the second coordinate. We introduce the orthogonal graph $G$ whose vertex set is given by the grid $S_x \times S_y$, and each vertex is linked by a directed edge to the closest nodes on the right and on the bottom (if any). 
            Consider now the set of vertices $S_M$ that belongs to the grid $S_x \times S_y$ and to the allocation region $A_M$. By monotonicity of $A_M$, $S_M$ is connected and respects a ``discrete'' monotonicity property: if $u$ is in $S_M$ and $v$ is a generic node of $G$ that dominates $u$, then $v$ is in $S_M$. Consider now the (Pareto) frontier of $S_M$: all the points $(u_1,u_2) \in S_M$ that do not strictly dominate other points in $S_M$. More precisely, $u \in S_M$ belongs to the frontier of $S_M$ if $\not \exists w \in S_M$ such that $u_1 > w_1$ and $u_2>w_2$.
            We prove that the path $\pi$ composed by the points on the frontier of $S_M$ is such that $M = M_{\pi}$ (note, $M_{\pi}$ is orthogonal by \Cref{prop:orthogonal}). In particular, we prove that $A_M = A_{M_{\pi}}$. 
            One inclusion is easy: $A_{\pi}$ contains $\supp(V) \cap A_M$ and therefore is contained in $A_M$ because of its minimality (\Cref{cl:minimal}). Consider now any $u \in A_{\pi}$, we want to prove that it belongs to $A_M$. Point $u$ is in $A_{\pi}$, therefore there exists $v \in \pi$ which is dominated by it, but $v$ also belong to $A_M$ and thus $u \in A_M$ by monotonicity.
        \end{proof}

        We are left with the final step of our proof, that concerns presenting an efficient algorithm to find $M^{\star}$: the revenue maximizing mechanism that is also orthogonal. First we observe that the support of $V$ is finite, therefore there exists a finite number of subsets of points that can be included in the allocation region. To each of these subsets we can associate the inclusion minimal mechanism, which is revenue maximizing and orthogonal (by Claims~\ref{cl:minimal} and \ref{cl:orthogonal}). 
        Consider now the orthogonal graph whose node set is composed by the Cartesian product of the $x$ and $y$ coordinates of {\em all} the points on the support of $V$ (and the edges are constructed similarly to what we did in the proof of \Cref{cl:orthogonal}); this graph will contain as subgraph the one used in the proof of \Cref{cl:orthogonal} for the optimal subset of the support of $V$, so it is enough to find the best orthogonal mechanism supported on such graph. By \Cref{prop:decomposition}, this is equivalent to finding the maximum weight complete path in $G$. This task can be done efficiently (i.e., in time polynomial in the cardinality of the support of $V$) by running the Belman-Ford algorithm that solves the shortest path problem on the same graph but where the costs are $c(e) = - w_V(e)$. In particular, one needs to compute the shortest path between any pair of north-east side points that can be the extreme of a complete path. Note, $G$ is acyclic so the algorithm converges to the optimum.
    \end{proof}

        \begin{figure*}[ht!]
	\captionsetup[subfigure]{aboveskip=0.5pt}
	\centering
	\begin{subfigure}{.33\textwidth}
		\centering
 		\scalebox{0.33}{\centering\includegraphics[width = 3.5\textwidth]{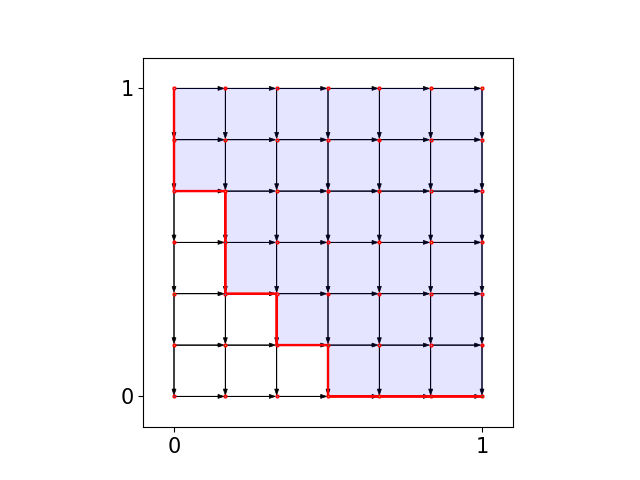}}
		\caption{\footnotesize An orthogonal graph: $G_{\nicefrac 16}$}
		\label{fig:grid1}
	\end{subfigure}\hspace{0.1pt}%
	\begin{subfigure}{.33\textwidth}
		\centering
		\scalebox{0.33}{\includegraphics[width = 3.5\textwidth]{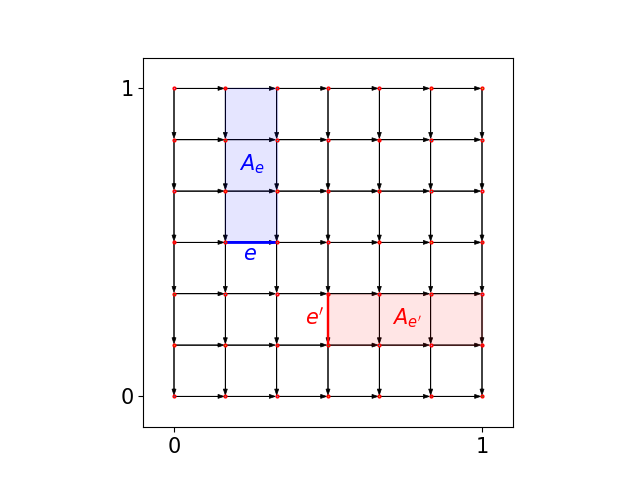}}
		\caption{\footnotesize Influence regions for $e$ and $e''$}
		\label{fig:grid2}
	\end{subfigure}\hspace{0.1pt}%
	\begin{subfigure}{.33\textwidth}
		\centering 		\scalebox{0.33}{\includegraphics[width = 3.5\textwidth]{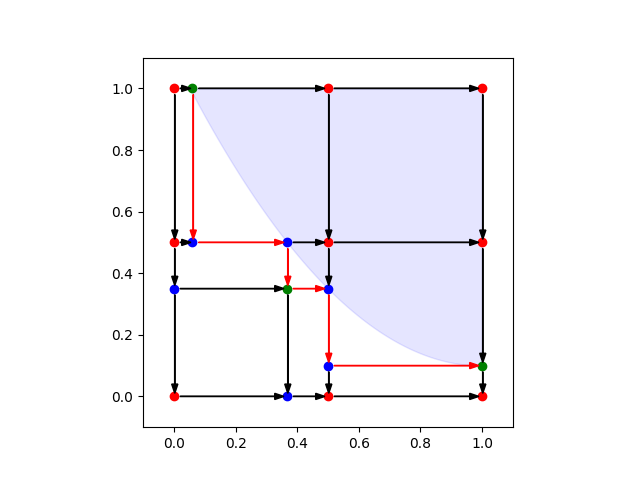}}
		\caption{\footnotesize Augmentation Procedure}
		\label{fig:augmented}
	\end{subfigure}
	\caption{\small An orthogonal graph (left) and the rectangles $A_e$ for vertical and horizontal edges (center). {Note that while the two figures are for a uniform grid, this is not needed for orthogonal mechanisms.} On the right, there is a visualization of the augmentation procedure: $G_{\nicefrac 12}$ (red nodes) is augmented with the three green points (note the auxiliary points in blue). These points have been chosen to mimic the procedure used in \Cref{lem:approx} to approximate the mechanism whose allocation region is shaded in blue with the orthogonal one corresponding to the red complete path.
 } 
\end{figure*}
    
\subsection{A (Learnable) Data-Dependent Discretization Scheme}
\label{sec:orthogonal}

    From an online learning perspective, \Cref{thm:best_mechanism} allows us to efficiently compute the best mechanism on the empirical distribution of the valuations seen so far in a Follow-The-Leader way. However, this would not yield a meaningful algorithm, as both the class $\M$ and even the subclass of orthogonal mechanisms $\M^{\perp}$ seem too complex for the revenue function to be learned efficiently (this is witnessed for instance, by the fact that $\M^{\perp}$ has unbounded pseudo-dimension, see \Cref{app:complexity}). As an intermediate step, we formally introduce a natural subfamily of the orthogonal graphs: the uniform grids (see also the second bullet after \Cref{ex:counter-example}). 
    \begin{definition}[Uniform Grid] \label{def:grid}
        For any $\eps > 0$, we define the uniform grid graph $\Geps = (V_{\eps},E_{\eps})$:
        \begin{itemize}
            \item $\Veps = \{0,{\eps}, 2 \eps, \ldots, 1\}^2$
            \item $\Eeps$ connects each point of the grid to its right and bottom neighbor (if any).
        \end{itemize}
    \end{definition}

    For instance, the graph depicted in \Cref{fig:grid1} is $G_{\nicefrac 16}.$ For any fixed $\eps$, the family $\M_{\eps}$ of all the mechanisms that are induced by complete paths in $G_{\eps}$ has the desirable property that its revenue function can be learned quickly, as we only need to learn the weight function $w_V$ on the $\Theta(\nicefrac 1{\eps^2})$ edges (by \Cref{prop:decomposition}). However, as we observed in \Cref{ex:counter-example}, for any fixed $\eps$, there exists a simple distribution $V$ such that the revenue of the best fixed mechanism in $\M_{\eps}$ is far from that of the best mechanism in $\M$.  
    We are {thus} in a situation where the large family of orthogonal mechanisms well approximates the optimal mechanism, but it is unclear whether it can be learned effectively. In contrast, mechanisms supported on uniform grids have revenue functions that are easy to learn but may exhibit a big gap with respect to the performance of the best mechanism overall. {To escape this situation, we next describe} 
    a class of mechanism that is {only} slightly larger than {the former class, and (approximately)} 
    retains the approximation property of the {latter class.}
    
    For any $\eps$, the grid $G_{\eps}$ partitions the $[0,1]^2$ square into $\nicefrac 1{\eps^2}$ tiles of area $\eps^2$ (out of simplicity, we often implicitly assume that $\nicefrac 1\eps$ is integer, this does not hinder the general validity of our results). Consider now one of these tiles, let's say the one whose top left vertex is $(i \eps, j \eps)$ and a point $(v_1,v_2)$ inside such tile. To avoid overlapping, we do not consider the south and west sides of each tile, so that $ i\eps < v_1 \le (i+1) \eps $ and $(j-1) \eps < v_2 \le j \eps$. We show a procedure to incorporate point $(v_1,v_2)$ into its tile (and thus into $G_{\eps}$) while retaining the orthogonality of the overall graph. We call the outcome of this procedure as $G_{\eps}$ {\em augmented} with point $(v_1,v_2)$.
    We have four cases:
    \begin{itemize}
        \item If $(v_1,v_2)$ belongs to the interior of the tile, i.e., $ i\eps < v_1 < (i+1) \eps $ and $(j-1) \eps < v_2 < j \eps$, then we add $(v_1,v_2)$ and four other nodes to the node set of of $G_{\eps}$: $(v_1,j \eps)$, $(v_1,(j-1) \eps)$, $(i \eps,v_2)$ and $((i+1) \eps,v_2)$. Then we add an edge between each one of these nodes and $(v_1,v_2)$, respecting the right-down direction. The four nodes added belong to the sides of the tile, so, to maintain orthogonality, we split such edges in two: for instance, let $e = (u,z)$ be the edge containing node $ w \in \{(v_1,j \eps), (v_1,(j-1) \eps), (i \eps,v_2),((i+1) \eps,v_2))$, then we remove $e$ from the edge set, and we insert $e' = (u,w)$ and $e'' = (w,z)$ in its place (south-west tile in \Cref{fig:augmented}).
        \item If $(v_1,v_2)$ belongs to a vertical edge $e = (u,w)$, then we only add two nodes: $(v_1,v_2)$ and $(i \eps,v_2)$, we split the corresponding vertical edges and add a horizontal edge between them (south-east tile in \Cref{fig:augmented}). 
        \item If $(v_1,v_2)$ belongs to a horizontal 
        edge, $e = (u,w)$, then we only add two nodes: $(v_1,v_2)$ and $(v_1,j \eps)$, we split the corresponding horizontal edges and add a vertical edge between them (north-west tile in \Cref{fig:augmented}).
        \item If $(v_1,v_2)$ is already a vertex of the graph, we do nothing.
    \end{itemize}
    We can apply this augmentation property recursively for a subset $S$ of points, at the condition that at most one node appears in any tile {of the original grid $G_{\eps}$}. Note, it is possible that edges of the original grid $G_{\eps}$ may be split in more than two edges, as, e.g., the $(0,\nicefrac 12)$-$(\nicefrac12, \nicefrac 12)$ edge in \Cref{fig:augmented}. We define the resulting orthogonal graph as the augmented grid graph $G_{\eps,S}$ and denote with $M_{\eps,S}$ the corresponding orthogonal mechanism. We then have the following definition:
    \begin{definition}[Augmented-Grid Mechanisms]
        For any $\eps > 0$, we define the family of Augmented-Grid Mechanisms $\hat {\M}_{\eps}$ as all the orthogonal mechanisms $M_{\eps,S}$ for some set of points $S$, with $|S| \le \nicefrac 2\eps$.  
    \end{definition}

    Surprisingly enough, adding only this small amount of points ($\Theta(\nicefrac 1\eps)$ many) gets the best of both the orthogonal and the grid mechanisms, achieving approximation and learnability! This is proved in the following Lemma. Starting from any mechanism $M$, it is possible to construct an augmented grid mechanism that well approximates $M$. The procedure is depicted in \Cref{fig:augmented} and consists of adapting the tiles of $G_{\eps}$ that contain the border of the allocation region to mimic the behavior of $A_M$.
    \begin{lemma}
    \label{lem:approx}
        Fix any $\eps > 0$ and any random variable $V$ in $[0,1]^2$ describing the agents' valuations. For any mechanism $M \in \M$, there exists an augmented-grid mechanism $\hat M \in \hat \M_{\eps}$ such that 
        \[
            \E{\rev({M})} \le \E{\rev({\hat M})} + 2\eps
        \]
    \end{lemma}
    \begin{proof}
        Let $M$ be any mechanism, and consider its allocation region $A_M$. If $A_M$ is the union of tiles of $G_{\eps}$, then there is nothing to prove, as $M$ is then orthogonal. We then focus on all the tiles of $G_{\eps}$ that lie at the boundary of $A_M$, i.e., that intersect $A_M$ but are not entirely contained in it; these are the tiles we augment. These tiles are at most $\nicefrac{2}{\eps}$, so the output of this procedure belongs to $\hat \M_{\eps}$, as requested in the statement.

        Consider any one of these tiles, and the intersections $u$ and $v$ of the boundary of the allocation region with the sides on the tile.\footnote{The intersection $I$ of the boundary of $A_M$ and the sides of the tile may contain more than two points when the tile and $A_M$ share part of their borders. We define points $u$ and $v$ to be the extremal ones}. We have various cases:
        \begin{itemize}
            \item If $u$ belongs to the north side and $v$ to the east side, then we augment the tile with $(u_1,v_2)$ (see south-west tile of \Cref{fig:augmented})
            \item If $u$ belongs to the north side and $v$ to the south side, then we augment the tile with $u$ (see north-west tile of \Cref{fig:augmented})
            \item If $u$ belongs to the east side and $v$ to the west one, then we augment the tile with $v$ (see south-east tile of \Cref{fig:augmented})
            \item Otherwise we do not augment the tile
        \end{itemize}
        Consider now the minimal complete path $\pi$ on the augmented graph with the property that $A_{\pi}$ contains $A_M$. As $A_{\pi}$ contains $A_M$, the orthogonal mechanism always trades when $M$ does it. Conversely, for any valuation $v \in A_M$ the projections on $\pi$ and the boundary of $A_M$ {\em belong to the same tile!}, therefore the prices proposed by the orthogonal mechanism are at most an additive $\eps$ term smaller than the ones charged by $A_M$. This concludes the proof.   
    \end{proof}
    
    For any $M \in \M$, we {refer to} the mechanism $\hat M$ of \Cref{lem:approx} as {\em the associated augmented-grid mechanism with precision $\eps$}. An immediate consequence of \Cref{lem:approx} is that the best mechanism in $\hat \M_{\eps}$ is only an additive $\Theta(\eps)$ factor away from the optimal mechanism. 

    \begin{theorem}
    \label{thm:approximation}
        Fix any $\eps > 0$ such that $\nicefrac{1}{\eps} \in \mathbb{N}$ and any random variable $V$ in $[0,1]^2$ describing the agents' valuations, then there exists $\hat M \in \hat \M_{\eps}$ such that 
        \[
            \sup_{M \in \M} \E{\rev({M})} - \E{\rev(\hat M)} \le 3\eps.
        \]
    \end{theorem}

    We move our attention to learnability. We exploit the edge-decomposition of the revenue function (\Cref{prop:decomposition}) to argue that, to achieve an $\eps$-precise approximation of the revenue function on all the mechanisms in  $\hat{\M}_{\eps}$, it is enough to achieve an $\eps^2$ approximation of all the rectangles in the $[0,1]^2$ square (as any path contains $O(\nicefrac 1 \eps)$ edges and each edge $e$ is associated to a rectangular region $A_e$). Since the range space of all the rectangles has constant VC dimension, $n \in \tilde \Omega(\nicefrac{1}{\eps^4})$ samples are sufficient.
    \begin{theorem}
    \label{thm:learnability}
        Fix any failure probability $\delta \in (0,\nicefrac 12)$ and precision $\eps > 0$, and let $V_1, V_2, \dots, V_n$ be $n$ i.i.d. samples from the random variable $V$ describing the agents valuations. If $n \ge \tfrac{256}{\eps^4}\left(32 \log \tfrac {8}{\eps} + \log \tfrac 2\delta \right)$, then the following holds with probability at least $1-\delta$: 
        \[
            |\hat{\mathbb E}[\rev(M)] - \E{\rev(M)}| \le \eps, \, \forall M \in \hat \M_{\eps},
        \]
        where $\hat{\mathbb E}$ denotes the expectation with respect to the empirical distribution on the $n$ samples.
    \end{theorem}
     \begin{proof}
        For any mechanism $\M \in  \hat \M_{\eps}$ let $\cE_{M}$ be the event that the empirical revenue of $M$ on the $n$ samples is $\eps$ close to the actual one:
        \[
            \cE_M = \{|\hat{\mathbb E}[\rev(M)] - \E{\rev(M)}| \le \eps\}.
        \]
        Let $\pi_M$ be the path in the grid-augmented graph associated with the boundary of $M$. Crucially, $\pi_M$ contains at most $\nicefrac 4\eps$ edges, this is because complete paths on $G_{\eps}$ have cardinality at most $\nicefrac 2 \eps$, and each added point adds at most $2$ edges each. 
        If we exploit the decomposition result of \Cref{prop:decomposition}, we can rewrite the difference between the two revenues as follows:
        \begin{align}
        \nonumber
            |\hat{\mathbb E}&[\rev(M)] - \E{\rev(M)}|  = |\sum_{e \in \pi_M} w_e (\hat{\mathbb{P}}(V \in A_e) - \P{V \in A_e})| \\
            &\le \sum_{e \in \pi_M} w_e |\hat{\mathbb{P}}(V \in A_e) - \P{V \in A_e}|
        \label{eq:hatrev}
            \le \frac{4}{\eps} \cdot \max_{e \in \pi_M} |\hat{\mathbb{P}}(V \in A_e) - \P{V \in A_e}|.
        \end{align}
        So the problem has reduced to bounding the error on the class of all the axis-aligned rectangles of the form $A_e$ for some $e.$ Such class has VC dimension of $4$, so by standard learning arguments (e.g., Theorem 14.15 of \citet{Mitzenmacher06}) and by our choice of $n$ we get that, with probability at least $1-\delta$, it holds $            \max_{e \in \pi_M} |\hat{\mathbb{P}}(V \in A_e) - \P{V \in A_e}| \le \frac {\eps^2} 4.$
        Plugging this into \Cref{eq:hatrev} yields the desired result.
    \end{proof}

\subsection{The \augment Algorithm}
\label{sec:adaptive}

    We have all the ingredients to present and analyze our learning algorithm: \augment. At a generic time $t$, it computes $M^\star_t$, the revenue-maximizing mechanism over the $t-1$ valuations observed so far (note, this can be done efficiently because of \Cref{thm:best_mechanism}) and then posts the associated grid-augmented mechanism $M^t$ for a suitable precision parameter $\eps_t$. We refer to the pseudocode for further details. Given $M^\star_t$ it is easy to compute the associated augmented mechanism by following the procedure described in \Cref{lem:approx} and visualized in \Cref{fig:augmented}.
    
\begin{algorithm}[t!]
    \begin{algorithmic}[ht]
        \State \textbf{Environment}: Fixed random variable $V$ from which valuations $\{(v_1^t,v_2^t)\}$ are drawn i.i.d.    
        \State  Propose any mechanism $M^1$ and observe $(v^1_1,v_2^1)$ 
        \For{time $t=2,\ldots,T$}
            \State Set precision parameter $\eps_t = 14 \sqrt[4]{\nicefrac{\log T}{t}}$
            \State Let $M^\star_t$ be the revenue-maximizing mechanism over valuations $(v^1_1,v_2^1), \dots, (v^{t-1}_1,v_2^{t-1})$
            \State Let $M^t \in \hat\M_{\eps}$ be the grid-augmented mechanism with precision $\eps$ associated with $M^\star_t$
            \State Propose mechanism $M^t$ and observe $(v^1_t,v_2^t)$
        \EndFor
    \end{algorithmic}
    \caption*{\textbf{\augment}}
    \label{alg:adaptive}
    \end{algorithm}

    \begin{theorem}\label{thm:stochastic upper}
    Consider the Repeated Joint Ads problem in the stochastic i.i.d. setting; then, there exists a learning algorithm $\A$ such that
            \[
            R_T(\A) \le 35 \sqrt[4]{T^3 \cdot \log T}.
            \]
    \end{theorem}
    \begin{proof}
        We prove that \augment yields the desired regret bound in the stochastic case. For all $t = 2, \dots, T$, we denote with $\cE_t$ the event that all the mechanisms in $\hat \M_{\eps_t}$ are well approximated by the first $t-1$ samples:
        \[
            \cE_t = 
            \left\{
                |\hat{\mathbb E}_{t-1}[\rev(M)] - \E{\rev(M)}| \le \eps_{t}, \,\forall M \in \hat \M_{\eps_t}
            \right\},
        \]
        where $\hat{\mathbb E}_{t-1}$ denotes the empirical expectation with respect to the first $t-1$ samples. By our choice of $\eps_t$ and \Cref{thm:learnability}, we have that the $\cE_t$ are realized with high probability, i.e., 
        \begin{equation}
            \label{eq:clean}
            \P{\mathcal E_t} \ge 1 - \nicefrac 1T.    
        \end{equation}

        Now focus on what happens at time $t$, conditioning on $\cE_t$ happening. We have:
        \begin{align*}
                \sup_{M \in \M} \E{\rev_t(M)} &\le \E{\rev_t(\hat M)} + 3 \eps_t \tag{for some $\hat M \in \hat \M_{\eps_t}$, by \Cref{thm:approximation}}\\
                &\le \hat{\mathbb E}_{t-1}[{\rev_t(\hat M)}] + 4\eps_t \tag{Conditioning on $\cE_t$}\\
                &\le \hat{\mathbb E}_{t-1}[{\rev_t(M_t^{\star})}] + 4\eps_t \tag{By optimality of $M_t^{\star}$ w.r.t. $\hat{\mathbb{E}}_t$}\\
                &\le \hat{\mathbb E}_{t-1}[{\rev_t(M^t)}] + 6\eps_t \tag{By definition of $M^t$ and \Cref{lem:approx}}\\
                &\le{\mathbb E}[{\rev_t(M^t)}] + 7\eps_t \tag{Conditioning on $\cE_t$}
            \end{align*}
            Combining this chain of inequality with the bound on the probability of $\cE_t$ in \Cref{eq:clean}, we get that the instantaneous regret suffered by \augment is at most:
            \[
                \sup_{M \in \M} \E{\rev_t(M)} - 
                \E{\rev_t(M^t)} \le \left(7 \eps_t + \frac 2T\right) \le 8 \eps_t.
            \]
            Overall, the regret is upper bounded by $8\sum_{t=1}^T \eps_t$, which verifies the bound in the statement. 
    \end{proof}

 \section[Adversarial Lower Bound]{Adversarial Lower Bound } 
    \label{sec:adversarial_lower}

    We prove a lower bound on the regret when the sequence of valuations is generated adversarially. In particular, we show a stronger result: for any $\eps$, no learning algorithm suffers sublinear $(2-\eps)$-regret.   
    \begin{theorem}\label{thm:lb-adversarial}
        For any constant $\eps \in (0,1)$ and learning algorithm $\mathcal A$, there exists an adversarial instance of the Repeated Joint Ads problem such that the following inequality holds:
        \[
            \sup_{M \in \M} \sum_{t=1}^T \revt(M) - (2-\eps)\sum_{t=1}^T\E{ \revt(M_t)} \ge \frac{\eps^2}8 T.
        \]
    \end{theorem}
    \begin{proof}
        Let $\delta$ and $\zeta$ be small positive parameters we set later. We prove this lower bound via Yao's minimax principle: we construct a randomized instance against which any deterministic algorithm suffers expected linear $(2-\eps)$ regret. The randomized instance is built incrementally, together with an auxiliary sequence $(a_t,b_t)$, as follows: $a_1 = \nicefrac {\delta} 3$, $b_1 = \nicefrac {2\delta} 3$, then for any time step $t$: 
        \begin{itemize}
            \item with probability $\zeta$, the valuations are $(v_1^t, v_2^t) = (b_t, 1)$, and the auxiliary sequence is updated as follows: $(a_{t+1}, b_{t+1}) = (b_t + \Delta_{t+1}, b_t + 2\Delta_{t+1})$, where $\Delta_{t+1} = \nicefrac{\delta}{3^{t+1}}$
            \item with the remaining probability $1-\zeta$, the valuations are $(v_1^t, v_2^t) = (a_t, \zeta)$, and the next pair in the auxiliary sequence is $(a_{t+1}, b_{t+1}) = (a_t - \Delta_{t+1}, a_t - 2\Delta_{t+1})$.
        \end{itemize}

        This random sequence induces a {\em well defined} sequence of valuations and exhibits the key property that {\em for any realization}, there exists a threshold $\tau$ that separates all the valuations of the form $(a_t,\zeta)$ and all those of the form $(b_t,1)$. Formally, denote with $R$ the time steps characterized by a $(a_t,\zeta)$ valuation, and with $L$ the ones corresponding to $(b_t,1)$ valuations. We have the following result, whose formal proof is deferred to 
        \Cref{app:adversarial lower bound}.
        \begin{restatable}{proposition}{properties}
        \label{prop:properties}
            For any realization of the valuations and auxiliary sequence, we have:
            \begin{itemize}
                \item[(i)] $a_t \in [0,\delta]$ and $b_t \in  [0,\delta]$, for all $t = 1,\ldots, T$.
                 \item[(ii)] there exists a value $\tau \in  (0,\delta)$ such that $b_t < \tau$ for all $t \in R$, while $a_t \ge \tau$ for all $t \in L$. 
            \end{itemize}
        \end{restatable}

        The second point of \Cref{prop:properties} implies the existence of an optimal mechanism $M^\star$ that extracts the maximum revenue possible at each time step (the sum of the valuations): it accepts all the bids of the form $(x,1)$ for $x \le \tau$, and all the bids of the form $(x,y)$ with $x > \tau$ and $y \ge \zeta.$
        In each iteration, the expected revenue that $M^\star$ extracts is at least $ \E{\rev_t(M^\star)} \ge \zeta \cdot 1 + (1-\zeta) \cdot \zeta = \zeta \cdot (2-\zeta)$.
        
        Consider now any deterministic algorithm and its expected revenue at time $t$. Fix the history of the sequence up to time $t$ (excluded); there are two options for the next valuation: it is either $(b_t,1)$ with probability $\zeta$, or $(a_t,\zeta)$ with the remaining probability. By design, it holds that the possible valuation $(b_t,1)$ would dominate the other possible valuation $(a_t,\zeta)$. therefore any mechanism that the learning algorithm may post can do one of two things: either make the trade happen for both possible valuations $(b_t,1)$ and $(a_t,\zeta)$ (with expected revenue of at most $(a_t + \zeta)$), or only for $(b_t,1)$ (with expected revenue of at most $\zeta(b_t + 1)$). All in all, the expected revenue is at most $\E{\revt(M^t)} \le \max\{(a_t + \zeta), \zeta(b_t + 1)\} \le \delta + \zeta$, where for the second inequality we used that $a_t,b_t \le \delta$ by point (i) of \Cref{prop:properties}, so that $(a_t + \zeta) \le \delta + \zeta$ and $\zeta(b_t + 1) \le  \zeta + \zeta \delta \le \zeta + \delta$.

        Summing up over all the time steps $t$, and setting $\zeta = \nicefrac \eps2$, and $\delta = \nicefrac{\eps^2}{16}$, the expected $(2-\eps)$-regret of any deterministic algorithm is at least $          \zeta (2-\zeta) T - (2-\eps) (\delta + \zeta) T \ge \frac{\eps^2}8 T$.
        \end{proof}

\section{The Smooth Adversary } 
    \label{sec:smooth_upper}

    In this section, we develop an efficient learning algorithm that achieves $O(T^{\nicefrac 23})$ regret against a $\sigma$-smooth adversary. First, we present a discretization result for the mechanisms supported on uniform grids. Second, we exploit this discretization to explicitly design our learning algorithm following a multiplicative-weight-update approach. Typically, such an approach would require the learner to maintain an array of weights, one for each expert/mechanism; this is, however, infeasible in our case due to the combinatorial nature of our mechanism space. We show how to sample from the desired distribution while only maintaining an exponentially smaller number of weights.

    \begin{figure}[t!]
    \centering
    	\includegraphics[width = 0.5 \textwidth]{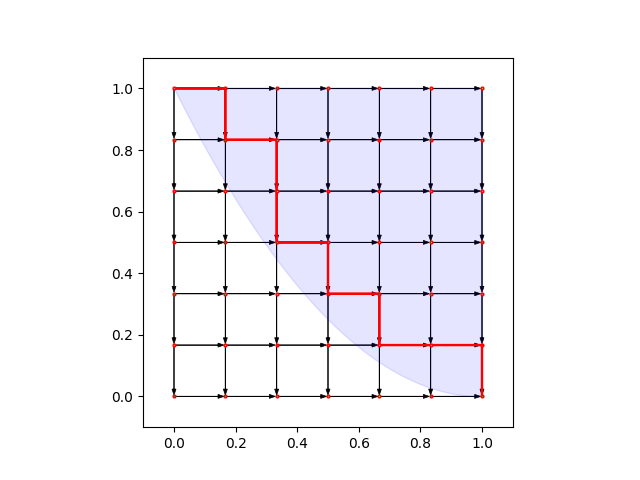}
    		\caption{\small Visualization of the grid $\Geps$ for $\eps = \nicefrac 16$. The shaded area represents the allocation region of some mechanism, while the red path is the corresponding inner hull.} 
    	\label{fig:inner}
    \end{figure}

     \subsection{The Discretization Lemma}
    We start by showing that against a $\sigma$-smooth adversary, it suffices to consider all mechanisms $\mathcal{M}_\eps$ on the uniform $G_{\eps}$ grid to obtain a small discretization error. This fact contrasts sharply with the case where we don't impose smoothness; see \Cref{ex:counter-example}. Given the smooth nature of the adversary, we can ignore what happens in regions of zero Lebesgue measure; therefore, we restrict our attention to complete paths in $G_{\eps}$ that start from $(0,1)$ and terminate in $(1,0)$.
    \begin{lemma}[Discretization error for $\sigma$-smooth adversaries]  \label{lem:discretization_smooth}
        For any $\eps > 0$, any $\sigma \in (0,1]$, and any sequence of $\sigma$-smooth distributions, the following inequality holds: 
        \[
            \sup_{M \in \mathcal{M}} \sum_{t=1}^{T} \E{\revt(M)} - \max_{M_{\eps} \in \mathcal{M}_{\eps}} \sum_{t=1}^{T} \E{\revt(M_{\eps})} \le \frac{5}{\sigma}\eps T. 
        \]
    \end{lemma}
    \begin{proof}
        Fix $\eps$, the sequence of $\sigma$-smooth distributions as in the statement, and consider any mechanism $\M^\star$ such that its total expected revenue is only at most an additive $\eps T$ far from the one of the $\sup$.
        We define the inner-hull $H$ of $M^\star$ as the set of all the tiles in the grid $G_{\eps}$ contained in the allocation region $A^\star$ of $M^\star$ (see \Cref{fig:inner}). Let $M$ be the mechanism in $\M_{\eps}$ whose allocation region is the inner-hull $H$ of $M^\star$.       
        For any time step $t$,  we have the following inequality:
        \begin{align*}
            \E{\revt(M^\star)} &=  \E{\revt(M^\star) \ind{(v_1^t,v_2^t) \in H}} +  \E{\revt(M^\star) \ind{(v_1^t,v_2^t) \in A^\star\setminus H}}\\
            &\le \E{\revt(M)} + 2\P{(v_1^t,v_2^t) \in A^\star\setminus H} \tag*{\text{($\revt(M^\star) \le 2$)}} \\ 
            &\le \E{\revt(M)} + \tfrac{2}{\sigma}\nu(A^\star\setminus H)\le  \E{\revt(M)} + 4\tfrac{\eps}{\sigma}. \tag*{\text{($\sigma$-smoothness)}}  
        \end{align*}
        Crucially, in the first inequality, we used that all the valuations in $H\subseteq A^\star$ result in trade under both mechanisms, but $M$ charges higher prices. Finally, the Lebesgue measure $\nu$ of $A^\star\setminus H$ is at most $2\eps$ because it is contained in at most $\frac{2}{\eps}$ tiles of $G_{\eps}$ (all the tiles whose interior have a nonempty intersection with the boundary of $A^\star$) of area $\eps^2$.
    \end{proof}

    \subsection{Path Learning Algorithm }
    
        The discretization result in \Cref{lem:discretization_smooth} and the bijection between the mechanisms $\M_{\eps}$ and the set of all complete paths $\mathcal P_{\eps}$ on $\Geps$, implies that it is enough for the learner to achieve sublinear regret with respect to the best mechanism in a large but {\em finite} class of candidates. We exploit this in our algorithm \pL, which is simply an instantiation of the well-known Hedge algorithm where each path in $\mathcal P_{\eps}$ is an expert. We refer to the pseudocode for further details and to \Cref{sec:sampling} for an efficient implementation of the sampling procedure. Before presenting the Theorem, which analyzes \pL, we observe that a simple combinatorial argument yields an upper bound on the number of mechanisms in $\M_{\eps}$. The formal proof is deferred to \Cref{app:smooth_upper}.
        \begin{lemma}
        \label{lem:N_e}
            For any $\eps > 0$, the number $N_{\eps}$ of mechanisms in $\M_{\eps}$ is at most $2\cdot 4^{\nicefrac1{\eps}}.$
        \end{lemma}
       \begin{proof}
            To show this, we consider the number of "down-right" paths
            in $\Geps$ from $(0,1)$ to $(1,0)$, as we know that there is a bijection between these paths and the mechanisms in $\M_{\eps}$. Each path can be described by specifying which of its $\nicefrac{2}{\eps}$ edges are the $\nicefrac{1}{\eps}$ vertical edges and which are the $\nicefrac{1}{\eps}$ horizontals. All in all: 
            \[
                N_{\eps} = \binom{\nicefrac{2}{\eps}}{\nicefrac{1}{\eps}} \le 2\cdot 4^{\nicefrac1{\eps}},
            \]
            where the inequality follows from Sterling's inequalities (e.g., Formula 9.15 of \citet{Feller67}).
        \end{proof}
        We now have all the ingredients to prove the regret bound provided by $\pL$. We present a way to efficiently implement the algorithm in \Cref{sec:sampling}
        \begin{theorem}\label{thm:hedge}
            Consider the Repeated Joint Ads problem against the $\sigma$-smooth adversary; then there exists a learning algorithm $\A$ such that
            \[
                R_T(\A) \le    \frac {13}{\sigma}\cdot T^{\nicefrac 23}.
            \]
        \end{theorem}
        \begin{proof}
            We prove that \pL with the right choice of the discretization and learning parameters verifies the statement of the Theorem. To this end, we combine \Cref{lem:discretization_smooth} with the well-known regret bound of Hedge  and optimize the two parameters:
            \begin{align*}
                R_T(\pL) &= \sup_{M \in \mathcal{M}} \E{\sum_{t=1}^{T} \revt(M) - \sum_{t=1}^T \revt(M^t)} \\
                &\le  \frac{5}{\sigma}\eps T + \max_{M_{\eps}\in \mathcal{M}_{\eps}}  \E{\sum_{t=1}^{T} \revt(M_{\eps}) - \sum_{t=1}^T \revt(M^t)}\\
                &\le \frac{5}{\sigma}\eps T +  4\sqrt{T \log N_{\eps}} \le  \frac{5}{\sigma}\eps T +  8\sqrt{\frac{T}{\eps}} \le \frac {13}{\sigma}\cdot T^{\nicefrac 23},
            \end{align*}
            where we used the standard bound of Hedge when the rewards are bounded in $[0,2]$ (see, e.g., Theorem 2.5 of \citet{AroraHK12} for $\eta = T^{- \nicefrac 12}$), \Cref{lem:N_e} and we set $\eps = T^{-\nicefrac 13}$.
        \end{proof}

        \begin{algorithm}[t!]
        \begin{algorithmic}[1]
            \State \textbf{Environment}: Sequence of $\sigma$-smooth random variables $\{(V_1^t,V_2^t)\}$
            \State \textbf{Input:} time horizon $T$, discretization parameter $\eps \in (0,1)$ and learning parameter $\eta \in (0,1)$
            \State For each path $\pi \in \mathcal P_{\eps}$ set $w^1_\pi = 1$ and $q^1_\pi = \frac{w^1_\pi}{W_1}$, where $W_1 = \sum_{\pi \in \mathcal P_{\eps}}w^1_\pi$. \label{line:weights_1}
            \For{each round $t=1,2,\dots,T$}
                \State Draw a path $\pi^t \in \mathcal P_{\eps}$ according to the distribution $q^t$\label{line:sampling}
                \State Implement the mechanism $M^t=M_{\pi^t}$ and observe $(v_1^t,v_2^t)\sim (V_1^t,V_2^t)$
                \State For each $\pi \in \mathcal P_{\eps}$ set $w_\pi^t = w_\pi^t \cdot \exp{[\eta \cdot  \revt({M_\pi})]}$ and $q^t_\pi = \frac{w^t_\pi}{W_t}$, where $W_t = \sum_{\pi \in \mathcal P_{\eps}}w^t_\pi$. \label{line:weights_2}
            \EndFor
        \end{algorithmic}
        \caption*{\textbf{\pL}}
        \end{algorithm}

    \subsection{Sampling Efficiently: From Paths to Edges }\label{sec:sampling}
    
        We show how to efficiently implement the sampling step of \pL (line \ref{line:sampling}) without the need to maintain exponentially many weights (lines \ref{line:weights_1} and \ref{line:weights_2}). We do this by constructing a sampling scheme that only maintains weights for the edges of $\Geps$, similarly to what is done, e.g., for online shortest path in \citet{TakimotoW03} and \citet{GyorgyLL05}. The crucial idea is to maintain conditional-probability weights on the edges so that it is possible to sample paths from the desired distribution (line \ref{line:sampling} of \pL) ``one edge at a time'' using these edge weights. 

        \begin{theorem}
        \label{thm:sampling}
            For any time step $t$  and sequence of valuations up to time $t$, excluded, there exist edge probabilities $q^t_e$ such that the following three properties hold: 
            \begin{itemize}
                \item[(i)] $q_e^t \in [0,1]$ for all $e$ edges in $\Geps$,
                \item[(ii)] $\sum_{e \in E_u^+}q_e^t = 1 $ for all node $u$ in $\Geps$, where $E_u^+$ are the outgoing edges from $u$
                \item[(iii)] $q_{\pi}^t = \prod_{e \in p}q^t_e$ for all paths $p \in \mathcal P_{\eps}$.
            \end{itemize}
            Furthermore, it is possible to compute $q_e^t$ in time polynomial in $T$ and $\nicefrac{1}{\eps}$.
        \end{theorem}

        Having access to the edge probabilities $q_e^t$, it is easy to sample from the distribution over paths as in line \ref{line:sampling} of \pL by constructing the path one edge at the time, using the $q_e^t$ to pick the following edge.

        \begin{proof}[Proof of Theorem~\ref{thm:sampling}]
            We prove the Theorem in three steps. First, we define the edge version of the weights in lines \ref{line:weights_1} and \ref{line:weights_2} of \pL. For each edge $e$ in $\Geps$, let $A_e$ be the allocation rectangle it ``covers'' (this is the influence region defined for general orthogonal graphs in \Cref{def:weights}): if $e$ is a vertical edge connecting $(\nicefrac{i}\eps,\nicefrac{j}\eps)$ to $(\nicefrac{i}\eps,\nicefrac{j-1}\eps)$, then $A_e$ is the rectangle $[\nicefrac{i} \eps,1]\times [\nicefrac{j-1}\eps,\nicefrac{j}\eps]$, similarly, if $e$ is a horizontal edge connecting $(\nicefrac{i}\eps,\nicefrac{j}\eps)$ to $(\nicefrac{i+1}\eps,\nicefrac{j}\eps)$, then $A_e$ is the rectangle $[\nicefrac{i}\eps,\nicefrac{i+1}\eps]\times [\nicefrac{j} \eps,1]$. Unlike in the general stochastic setting, the boundaries of the $A_e$ rectangles are irrelevant (as they carry no mass in smooth distributions). We can define the following edge weights: 
            \[
            w_e^t = \begin{cases}
                \exp{\{\eta \cdot \tfrac{i} \eps \sum_{s=1}^{t-1} \ind{(v^s_1,v^s_2) \in A_e}\}} \text{ if $e$ connects $(\tfrac{i}\eps,\tfrac{j}\eps)$ to $(\tfrac{i}\eps,\tfrac{j-1}\eps)$}
                \\
                \exp{\{\eta \cdot \tfrac{j} \eps \sum_{s=1}^{t-1} \ind{(v^s_1,v^s_2) \in A_e}\}}\text{ if $e$ connects $(\tfrac{i}\eps,\tfrac{j}\eps)$ to $(\tfrac{i+1}\eps,\tfrac{j}\eps)$}
            \end{cases}
        \]
        These weights are related to the ones used in \augment but have a different meaning. We also highlight that it is possible to compute these weights efficiently. In particular, by exploiting the structure of the revenue function and the exponential nature of weights (both for paths and edges), we can derive a simple relationship between $w^t_e$ and the path-weights $w_{\pi}^t$ used by \pL. The formal proof is deferred to 
        \Cref{app:smooth_upper}.
                
        \begin{restatable}{claim}{productdecomposition}
        \label{cl:prod}
            For any path $\pi \in \mathcal P_{\eps}$ the following equality holds: $w_{\pi}^t = \prod_{e \in \pi} w_e^t$
        \end{restatable}
        
        Second, for every node $u = (\nicefrac{i}\eps,\nicefrac{j}\eps)$ of the graph $\Geps$ let $N^+_u$ be its out-neighborhood. We define $w_{(1,0)}^t = 1$, and recursively the weights of the remaining nodes as follows:  $w_u^t = \sum_{v \in N^+_u} w_{(u,v)} w_v^t.$
        
        Note that the weights $w_u^t$ can be efficiently computed via dynamic programming starting from $w_{(1,0)}^t$. A node's weight equals the sum of the weights of all the paths stemming from it. This property, formalized by the following claim whose proof is deferred to \Cref{app:smooth_upper}, is crucial in defining the edge probabilities $q_e^t$.
        \begin{restatable}{claim}{sumprod}\label{cl:node}
            For any node $u\neq (1,0)$, let $\mathcal P^u_{\eps}$ denote the $u$-$(1,0)$ paths, then $w_u^t = \sum_{\pi \in \mathcal P^u_{\eps}} \prod_{e \in \pi} w^t_e$.
        \end{restatable}
        
        We are ready to define the edge probabilities $q_e^t$. Consider any node $u$; its outgoing edges have the following probabilities: if $N^+_u = \{v\}$, then $q_{(u,v)}^t = 1$, else $N^+_u = \{d,r\}$ and the probabilities are
        \[
            q_{(u,d)}^t = \tfrac{w_{(u,d)}^t \cdot w_{d}^t}{w_{(u,d)}^t \cdot w_{r}^t + w_{(u,r)}^t \cdot w_{r}^t} \text{ and }  q_{(u,r)}^t = \tfrac{w_{(u,r)}^t \cdot w_{r}^t}{w_{(u,d)}^t \cdot w_{r}^t + w_{(u,r)}^t \cdot w_{r}^t}.
        \]
        Now, properties (i) and (ii) in \Cref{thm:sampling} are clearly respected by $q_e^t$; we need to make sure that also the last property holds. For any $\pi = (e_1,e_2, \dots, e_{k})$, it holds that
        \[
             \P{\pi^t = \pi} = \prod_{i=1}^{k}\P{(e_1,e_2, \dots, e_{i}) \in \pi^t \mid (e_1,e_2, \dots, e_{i-1}) \in \pi^t}. 
        \]
        It then suffices to show that for any path $\pi$ and any $i=1, \dots, k$ the following equality holds
        \begin{equation}
        \label{eq:sampling}
            \P{(e_1,e_2, \dots, e_{i}) \in \pi^t \mid (e_1,e_2, \dots, e_{i-1}) \in \pi^t} = q_{e_i}^t.
        \end{equation}
        \Cref{eq:sampling} follows by a simple case analysis which is deferred to \Cref{cl:sampling} in \Cref{app:smooth_upper}.
        \end{proof}

\section{A \texorpdfstring{$\sqrt{T}$}{sqrt T} Lower Bound for the Smooth and the Stochastic Adversaries}
\label{sec:smooth_lower}
    
    We construct a hard instance for the Repeated Joint Ads Problem, which forces any learner to suffer at least $\Omega(\sqrt T)$ regret. Such instance is based on a smooth distribution from which valuations are drawn i.i.d., thus yielding a lower bound for both the smooth and the stochastic i.i.d. setting. 
    
    \begin{figure}
    \centering
    \includegraphics[width=0.4\textwidth]{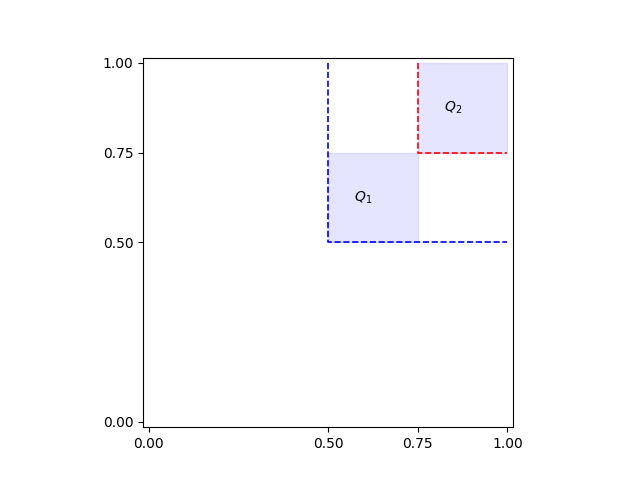}
    \caption{\small The two squares support the family of distributions we use in the lower bound. In blue, respectively red, are reported the boundary of the allocation region of $M_1$, respectively $M_2$.}
    \label{fig:lower_smooth}
    \end{figure}

        We introduce a family of smooth distributions, parameterized by $\alpha \in (\nicefrac{4}{15}, \nicefrac{2}{5})$. For the sake of simplicity, we denote with $V = (V_1, V_2)$ the random variables describing the agents' valuations, and with $\Pb^{\alpha}$ the probability measure (and with $\mathbb E^{\alpha}$ the expectation) under which $V$ has the following density: 
        \[
            f_{\alpha}(v_1,v_2) = 16 \left( \alpha \ind{(v_1,v_2) \in Q_1} + (1-\alpha) \ind{(v_1,v_2) \in Q_2}\right), 
        \]
        where $Q_1 = [\nicefrac 12,\nicefrac 34] \times  [\nicefrac 12,\nicefrac 34]$ and $Q_2 = [\nicefrac 34,1] \times  [\nicefrac 34,1]$ (see also \Cref{fig:lower_smooth}). In other words, under $\Pb^{\alpha}$ the random variable $V$ is drawn from a mixture of two uniform distributions, one supported on $Q_1$ (with probability $\alpha$) and one on $Q_2$ (with the remaining probability $1-\alpha$). A simple calculation shows the smoothness of the family, in particular, for any $\alpha \in (\nicefrac{4}{15}, \nicefrac{2}{5})$, $\Pb^{\alpha}$ is $\nicefrac{1}{12}$-smooth. We introduce two mechanisms: 
        \begin{itemize}
            \item Mechanism $M_1$, which allocates in the $[\nicefrac 12,1] \times [\nicefrac 12,1]$ square with $\Ea{\rev(M_1)} = 1$
            \item Mechanism $M_2$, which allocates in the $[\nicefrac 34,1] \times [\nicefrac 34,1]$ square with $\Ea{\rev(M_2)} = \tfrac 32(1-\alpha)$.
        \end{itemize}
        Specifically, it is possible to prove that for any $\Pb^{\alpha}$, either $M_1$ or $M_2$ are optimal, as they always dominate the other mechanisms. We present the formal result in the following Lemma, whose proof (as the complete one of the main Theorem) is deferred to  \Cref{app:smooth_lower}.
        \begin{restatable}[Domination]{lemma}{Domination}
        \label{lem:domination}
            Fix any $\alpha \in (\nicefrac{4}{15}, \nicefrac{2}{5})$ and any mechanism $M$, and let $A_M$ be its allocation region. Then, the following properties hold
            \begin{itemize}
                \item If $A_M \cap Q_1 \neq \emptyset$,  then $\Ea{\rev(M_1)} \ge \Ea{\rev(M)}$.
                \item If $A_M \cap Q_1 = \emptyset$,  then $\Ea{\rev(M_2)} \ge \Ea{\rev(M)}$.
            \end{itemize}
        \end{restatable}
        \begin{restatable}{theorem}{lowersmooth}\label{thm:lower_smooth}
        Consider the Repeated Joint Ads problem against the $\sigma$-smooth adversary, for $\sigma \in (0,\nicefrac 1{12})$, then any learning algorithm $\A$ suffers regret $R_T(\A) \ge \tfrac{3}{64}\sqrt{T}.$
        \end{restatable}
        
        \begin{proof}[Proof Sketch]
            Let $\eps$ be a small parameter to be set later, and let $\alpha_0= \nicefrac 13$. We consider two probability distributions from the family introduced above, corresponding to parameters $\alpha_0 + \eps$ and $\alpha_0 - \eps$. Consider an adversary that selects one of these two distributions uniformly at random and extracts the valuations i.i.d. from it. 
            We sketch here an explanation of {why any learning algorithm suffers regret $\Omega(\sqrt{T})$ against this adversary,} 
            and defer the formal proof to \Cref{app:smooth_lower}. \Cref{lem:domination} critically tells us that we can restrict our attention, without loss of generality, to learning algorithms that only play the two mechanisms $M_1$ and $M_2$ described above. $M_1$ is optimal for the first distribution, while $M_2$ is optimal on the second. Moreover, their revenue is an additive $\Theta(\eps)$ away from each other, so that posting the wrong mechanism incurs an instantaneous regret of $\Theta (\eps)$. Any learning algorithm needs $\Omega(\nicefrac 1{\eps^2})$ samples to establish which one of the two instances it is playing against and suffers loss $\Theta(\eps)$ until it discovers it. If $\eps$ is too small, the learner cannot discover the best mechanism with high probability, thus suffering a regret of $\Omega(\eps T)$. All in all, it suffers $\Omega(\min\{\nicefrac{1}{\eps},\eps T\})$ regret. Setting $\eps = \nicefrac{1}{\sqrt{T}}$ yields the desired result.
         \end{proof}

        \begin{corollary}
        \label{cor:lower_stochastic}
        Any learning algorithm suffers $\Omega(\sqrt{T})$ regret in the stochastic i.i.d. setting.
        \end{corollary}

\section{Conclusion}
Motivated by online retail, in this paper, we introduced the online version of the Joint Ads Problem, in which two buyers (e.g., a merchant and a brand) cooperate to secure a non-excludable good (e.g., an ad slot). This problem is similarly fundamental to non-excludable mechanism design, as the single-buyer single-item problem is to excludable mechanism design. We have studied the problem in various data-generation models, providing efficient learning algorithms for the stochastic and the $\sigma$-smooth framework. We complemented these positive results with various lower bounds, in particular proving that in the adversarial setting, it is impossible to achieve sublinear regret. 
    Many interesting questions remain, including settling the minimax regret in the stochastic and $\sigma$-smooth models. Additionally, it would be interesting to see if our adaptive discretization technique can be useful for other related problems. Finally, we hope our work can be a stepping stone for future investigation that extends our results to more buyers (bigger groups and multiple groups). 

\section*{Acknowledgments}
    The authors would like to thank an anonymous reviewer for pointing out a flaw in a previous version of the paper.
    The work of Federico Fusco is supported by the ERC Advanced Grant 788893 AMDROMA, the MIUR PRIN grant 2017R9FHSR ALGADIMAR, the MUR PRIN grant 2022EKNE5K ``Learning in Markets and Society'', the PNRR Project: ``SoBigData.it -Strengthening the Italian RI for Social Mining and Big Data Analytics'' (Prot. IR0000013 - Avviso n. 3264 del 28/12/2021), and the FAIR (Future Artificial Intelligence Research) project PE0000013, spoke 5.

\bibliographystyle{plainnat}
\bibliography{bibliography}

\clearpage
\appendix

\section{Supplementary Material}

\subsection[The sample complexity of M]{The Learning complexity of $\M$ and $\M^{\perp}$ }
    \label{app:complexity}

        We prove that for any $\eps \in (0,1)$, the pseudo-dimension of $\M_{\eps} \subseteq \M^{\perp}$ is $\Omega(\nicefrac 1\eps)$. This result implies that the class of all mechanisms has unbounded pseudo-dimension and that the same holds for $\M^{\perp}$. We start by recalling the definition of pseudo-dimension of a class of functions. Let $\Q$ be a set and consider a class $\F$ of functions $f: \Q \to[0,H]$. Let $S$ be $m$ points in $\Q$, labeled according to $f$, meaning that for each sample $x \in S$, $f(x)$ is known for all $f \in \F$. Let $(r_1,r_2, \dots, r_m) \in [0,H]^m$ be a set of targets for $S$. We say that $(r_1,r_2, \dots, r_m)$ witnesses the shattering of $S$ by $\F$ if, for each $T \subseteq S$, there exists some $f_T \in \F$ such that $f_T(x_i) \ge r_i$ for all $x_i \in T$ and $f_T(x_i)<r_i$ for all $x_i \in S \setminus T$. If there exists some $(r_1,r_2, \dots, r_m)$ witnessing the shattering of $S$, we say that $S$ is shatterable by $\F$. 
            
        \begin{definition}
            The pseudo-dimension of $\F$ is the size of the largest set $S$ shatterable by $\F$.
        \end{definition}
        
        In our setting, the set $\Q$ is $[0,1]^2$, $H=2$ and $\F$ is the set of all functions $f = f_M$ which map valuations to revenue according to some mechanism $M \in \M.$ For simplicity, we refer to the pseudo-dimension of $\M$ (or some subset of $\M$) without specifying that we are actually studying the complexity of the class of functions just described.

        \begin{theorem}
        \label{thm:pseudo}
            For any $\eps \in (0,1)$, such that $\nicefrac 1 \eps \in \mathbb N$, the pseudo-dimension of $M_{\varepsilon}$ is $\Omega(\nicefrac 1{\varepsilon})$.
        \end{theorem}
        \begin{proof}
            Consider $S = \{(i\varepsilon, 1- i\varepsilon) \in [0,1]^2\mid i = 1, 2, \dots \nicefrac 1 \varepsilon\}$, we show that these $\nicefrac 1 \eps$ points are shattered by $\M_\eps$, using the vector $(1,1,\dots,1) \in [0,2]^{\nicefrac 1 \eps}$ has witness. For any subset $T$ of these points, it is easy to find a path on $G_{\eps}$ which passes through all the nodes in $T$ (thus ensuring a revenue of $1$ and clearing the witness condition on these points), and whose allocation region does not contain the vertices in $S \setminus T$ (thus getting no revenue from them and clearing the witness condition on these points). We refer to \Cref{fig:pseudo} for a visualization.
        \end{proof}
        
        We have two immediate corollaries.
        
        \begin{corollary} 
            The pseudo-dimension of $\M^{\perp}$ and $\M$ is unbounded.
        \end{corollary}
        \begin{corollary}
            The VC dimension of all the monotone regions is unbounded.
        \end{corollary}
        
        \begin{figure}[ht!]
        \centering
        \includegraphics[width = 0.5 \textwidth]{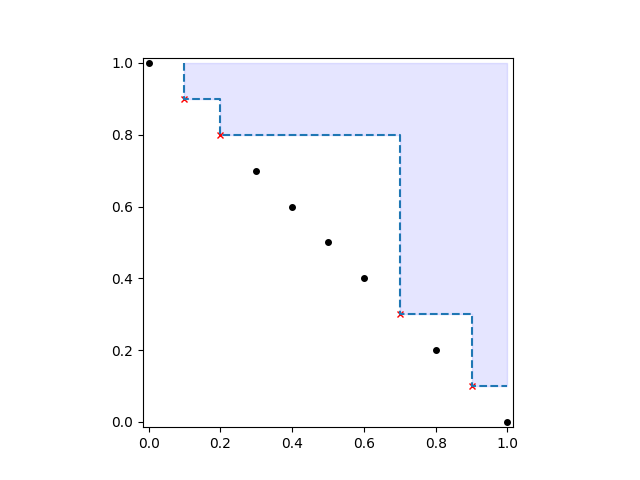}
        \caption{\small Example of the mechanism shattering points on the diagonal. The red points correspond to $T$ and yield revenue $1$ in the mechanisms corresponding to the shaded allocation region. All the black points in $S \setminus T$ are {\em not} in the allocation region and thus generate $0$ revenue.}
        \label{fig:pseudo}
    \end{figure}
    



\subsection[Proofs Omitted from Section 4]{Proofs Omitted from \Cref{sec:adversarial_lower} }
\label{app:adversarial lower bound}

    \properties*
    \begin{proof}
    We start by deriving closed-form solutions for $a_t$ and $b_t$, parameterized by the sets $L$ and $R$, which correspond to the coin tosses that lead to the two cases in constructing the input sequence. That is, $t \in R$ with probability $\zeta$ (i.e., $(v_1^t,v_2^t) = (b_t,1)$) and $t \in L$ with probability $1-\zeta$ (i.e., $(v_1^t,v_2^t) = (a_t,\zeta)$). 
    We say that $t = 0$ belongs to $R$ for convenience.
    Define
    \begin{align*}
    \alpha_t = \begin{cases}
    2 & \text{if $t \in R$} \\
    1 & \text{if $t \in L$}\\
    \end{cases}
    \quad 
    \text{and}
    \quad
    \gamma_t = \begin{cases}
    +1 & \text{if ${t-1} \in R$}\\
    -1 & \text{if $t-1 \in L$}
    \end{cases}.
    \end{align*}

    \begin{claim}\label{cla:sequences-closed-form}
    For any realization of the coin tosses, we have the following closed-form
    \begin{align*}
    &a_t = \gamma_t \cdot \Delta_t + \sum_{j=1}^{t-1} \gamma_j \cdot \alpha_j \cdot \Delta_j 
    \quad \text{and} \quad  
    b_t = 2\cdot \gamma_t \cdot \Delta_t + \sum_{j=1}^{t-1} \gamma_j \cdot \alpha_j \cdot \Delta_j.
    \end{align*}
    \end{claim}
    \begin{proof}[Proof of \Cref{cla:sequences-closed-form}]
    We prove the claim by induction. For $t = 1$ we have
    \begin{align*}
    &a_1 = \gamma_1 \cdot \Delta_1 = \Delta_1 = \tfrac{\delta}{3}
    \quad \text{and} \quad
    b_1 = 2 \cdot \gamma_1 \cdot  \Delta_1 = 2 \cdot \Delta_1 = \tfrac{2\delta}{3},
    \end{align*}
    as needed. For the inductive step, assume the claim is true up to $t$. We distinguish the two cases according to whether $t \in R$ or $L$. 
    If $t \in R$. Then, by the induction hypothesis, we have:
    \begin{align*}
    a_{t+1} &= b_t + \Delta_{t+1} 
    = \left[\left(\sum_{j=1}^{t-1} \gamma_j \cdot \alpha_j \cdot \Delta_j \right) +  \gamma_t \cdot 2 \cdot \Delta_t\right] + \Delta_{t+1}\\
    &= \left[\left(\sum_{j=1}^{t-1} \gamma_j \cdot \alpha_j \cdot \Delta_j \right) +  \gamma_t \cdot \alpha_t \cdot \Delta_t\right] + \gamma_{t+1} \cdot \Delta_{t+1}= \left(\sum_{j=1}^{t} \gamma_j \cdot \alpha_j \cdot \Delta_j \right) + \gamma_{t+1} \cdot 
    \Delta_{t+1}
    \end{align*}
and similarly
\begin{align*}
    b_{t+1} &= b_t + 2 \cdot  \Delta_{t+1} 
    = \left[\left(\sum_{j=1}^{t-1} \gamma_j \cdot \alpha_j \cdot \Delta_j \right) +  \gamma_t \cdot 2 \cdot \Delta_t\right] + 2 \cdot \Delta_{t+1}\\
    &= \left[\left(\sum_{j=1}^{t-1} \gamma_j \cdot \alpha_j \cdot \Delta_j \right) +  \gamma_t \cdot \alpha_t \cdot \Delta_t\right] + \gamma_{t+1} \cdot 2 \cdot \Delta_{t+1}= \left(\sum_{j=1}^{t} \gamma_j \cdot \alpha_j \cdot \Delta_j \right) + \gamma_{t+1} \cdot 2 \cdot 
    \Delta_{t+1}.
    \end{align*}
    The remaining case, i.e., when $t \in L$ is analogous.
    \end{proof}

    The closed formulas in Claim~\ref{cla:sequences-closed-form} are all we need to prove point (i) of the Proposition: for all $t$:
    \begin{align*}
    a_t = \left(\sum_{j=1}^{t-1} \gamma_j \cdot \alpha_j \cdot \Delta_j \right) + \gamma_t \cdot \Delta_t 
    \leq \sum_{j=1}^{T} 2 \cdot \Delta_j \leq 2 \sum_{j=1}^{\infty} \frac{\delta}{3^j} 
    = \delta.
    \end{align*}
    On the other hand, we have that
    \begin{align*}
    a_t 
    = \left(\sum_{j=1}^{t-1} \gamma_j \cdot \alpha_j \cdot \Delta_j \right) + \gamma_t \cdot \Delta_t 
    \geq \Delta_1 - \sum_{j=2}^{\infty} 2 \cdot \Delta_j = \frac{\delta}{3} - 2 \left(\bigg(\sum_{j=1}^{\infty} \frac{\delta}{3^j}\bigg) - \frac{\delta}{3}\right) 
    = 0.
    \end{align*}
    
    Similarly, again by Claim~\ref{cla:sequences-closed-form}, for all $t$, it is possible to verify that $b_t \in (0,\delta)$. We omit the calculations that are analogous. 
    as claimed. To address point (ii) of the Proposition, we exploit a natural notion of monotonicity. Formally, we have the following result. 

    \begin{claim}\label{cla:sequences-are-monotone}
    The sequence $a_t$ is decreasing for $t\in R$, while $b_t$ is increasing for $t \in L$.
    \end{claim}
    \begin{proof}[Proof of \Cref{cla:sequences-are-monotone}]
    We prove the result for $t \in R$; the other case is analogous. Fix any $t$ and $t'$ with $t<t'$ both in $R$, such that no intermediate time step is in $R$, we have the following: 
    \[
    b_{t'} - b_{t} = 2 \Delta_{t+1} - \left( \sum_{j=t+2}^{t'} 2 \Delta_{j+1} \right) = 2\delta \left(\frac{1}{3^{t+1}} - \sum_{j=t+2}^{t'} \frac{1}{3^{j}}\right) \geq \frac{2\delta}{3^{t+1}} \left(1- \sum_{j=1}^{\infty} \frac{1}{3^j}\right) = \frac{\delta}{3^{t+1}} > 0. \qedhere
    \]
    \end{proof}

    To show that there exists a threshold $\tau \in (0,1)$ that separates the two sequences, it suffices to show that for $t',t''$ such that $t'$ is the last index in $L$, and $t''$ is the last one in $R$,
    it holds that $b_{t''} < a_{t'}$. This property is sufficient because then we can choose $\tau$ halfway between $b_{t''}$ and $a_{t'}$ to ensure that $\tau \in (0,\delta)$ (by point (i) of the Proposition), while for all $(a_t,\zeta)$ resp.~$(b_t,1)$ in the sequence it will hold that $a_t \geq a_{t''} > \tau$ resp.~$b_t < b_{t''} < \tau$ (by Claim~\ref{cla:sequences-are-monotone}).
    
    \paragraph{\bf Case 1:} $t' < t''$. In this case, by the definition of $t'$ and $t''$, we have that
    \begin{align*}
    b_{t''} &= a_{t'} - 2\Delta_{t+1} + \sum_{j = t+2}^{T} 2\Delta_{j}= a_{t'} - 2\frac{\delta}{3^{t+1}} + \sum_{j = t+2}^{T} 2\frac{\delta}{3^j}\\
    &< a_{t'} - 2\frac{\delta}{3^{t+1}} + \sum_{j = t+2}^{\infty} 2\frac{\delta}{3^j}=a_{t'} - 2\frac{\delta}{3^{t+1}} \left(1 - \sum_{j = 1}^{\infty} \frac{1}{3^j}\right)= a_{t'} - \frac{\delta}{3^{t+1}} < a_{t'}.
    \end{align*}
    
    \paragraph{\bf Case 2:} $t' > t''$. In this case, we can again use the definition of $t'$ and $t''$ to conclude that
    \begin{align*}
    a_{t'} &= b_{t''} + \Delta_{t+1} - \sum_{j = t+2}^{T} \Delta_{j}= b_{t'} + \frac{\delta}{3^{t+1}} - \sum_{j = t+2}^{T} \frac{\delta}{3^j}\\
    &> b_{t''} + \frac{\delta}{3^{t+1}} - \sum_{j = t+2}^{\infty} \frac{\delta}{3^j}=b_{t''} + \frac{\delta}{3^{t+1}} \left(1 - \sum_{j = 1}^{\infty} \frac{1}{3^j}\right)= b_{t''} + \frac{\delta}{2 \cdot 3^{t+1}} > b_{t''}. \qedhere
    \end{align*}
    \end{proof}

    \subsection[Proofs Omitted From Section 5]{Proofs Omitted from \Cref{sec:smooth_upper}}
    \label{app:smooth_upper}

    \productdecomposition*
        \begin{proof}[Proof of \Cref{cl:prod}]
        Fix any sequence of valuations $(v_1^1,v_2^1), \dots (v_1^t,v_2^t)$ and any path $\pi \in \mathcal P_{\eps}$ (with $M = M_{\pi}$ as its associated mechanism), we prove the following equality, which implies the statement of the claim:
        \[
            \sum_{s=1}^t \rev_s(M) = \sum_{s=1}^t \sum_{e \in \pi} w_e \ind{(v^s_1,v^s_2) \in A_e},
        \]
        where $w_e$ is the intrinsic weight defined in \Cref{def:weights} and stands for $\nicefrac{i} \eps$ if $e$ connects $(\nicefrac{i}\eps,\nicefrac{j}\eps)$ to $(\nicefrac{i}\eps,\nicefrac{j-1}\eps)$ or $(\nicefrac{j}\eps,\nicefrac{i}\eps)$ to $(\nicefrac{j+1}\eps,\nicefrac{i}\eps)$.
        In particular, it is enough to prove that, for any time step $s$, the following hold:
        \begin{equation}
        \label{eq:path_to_edges}
            \rev_s(M) = \sum_{e \in p} w_e \ind{(v^s_1,v^s_2) \in A_e}.
        \end{equation}
        To see why this equality holds, we consider two cases, according to whether a trade happened or not. If $(v_1^s,v_2^s)$ does not lie in the allocation region of $M$, then it means there is no rectangle $A_e$ with $e \in \pi$ that contains it: therefore, the two sides of \Cref{eq:path_to_edges} are both $0.$ If $(v_1^s,v_2^s)$ induces a trade, then it belongs to precisely two rectangles $A_e$ and $A_{e'}$ for $e, e'$ (with $e$ horizontal and $e'$ vertical, note this holds because the realized valuations belong to the edges of $G_{\eps}$ with zero probability) in $\pi$; in particular, the payments of the agents for the mechanisms are exactly $w_e$ (for the second agents) and $w_{e'}$ (for the first agent). 
    \end{proof}

    \sumprod*
        \begin{proof}[Proof of \Cref{cl:node}]
        We prove this result by induction on the distance $d$ of $u$ to the sink $(1,0)$. We start with the base case $d=1$ for which the equality clearly holds: there are only two nodes to consider there, $(1-\nicefrac 1\eps, 0)$ and $(1, \nicefrac 1\eps)$ and only one edge connecting them to $(1,0)$ (recall the weight of the node $(1,0)$ is set to $1$).

        We now consider the general case and let $N_u^+$ be its out-neighborhood and $E_u^+$ the edge(s) that connect(s) $u$ to its neighbor(s). The paths in $\mathcal P_{\eps}^u$ can be divided according to their first edge:
        \begin{align}
            w_u^t &= \sum_{v \in  N_u^+} w_{(u,v)} w_v^t 
            = \sum_{v \in  N_u^+} w_{(u,v)} \sum_{\pi' \in \mathcal P_{\eps}^v} \prod_{e \in \pi'} w_e^t  
            = \sum_{\pi \in \mathcal P_{\eps}^u} \prod_{e \in \pi} w_e^t,
        \end{align}
        where the first equality is due to the $w_u^t$ definition, and the second to the induction hypothesis.
    \end{proof}

\begin{claim}
        \label{cl:sampling}
            It holds that 
            \(
            \P{(e_1,e_2, \dots, e_{i}) \in \pi^t \mid (e_1,e_2, \dots, e_{i-1}) \in \pi^t} = q_{e_i}^t.
            \)
        \end{claim}        \begin{proof}
        We have two cases: if $e_i$ is the only possible edge after $e_{i-1}$, i.e., the endpoint of $e_{i-1}$ has only one out-neighbor, then \Cref{eq:sampling} holds as both the sides of the formula are $1.$ Consider now the other case. Let $u$ be the end-point of edge $e_{i-1}$, $d$ and $r$ the two out-neighbors of $u$ and assume without loss of generality that $e_i = (u,r)$. To simplify notation, let $\mathcal P_{:i}$ be the set of all paths that coincide with $(e_1,\dots, e_{k})$ up to its $i^{th}$ edge, included.
        We have the following chain of equalities:
        \begin{align*}
            \mathbb P ((e_1,e_2, \dots, e_{i}) \in \pi^t \mid (e_1,e_2, \dots, e_{i-1}) \in \pi^t) &= \frac{\P{(e_1,e_2, \dots, e_{i}) \in \pi^t}}{\P{(e_1,e_2, \dots, e_{i-1}) \in \pi^t}}\\
            &= \frac{\sum_{\pi\in \mathcal P_{:i}}q^t_{\pi}}{\sum_{\pi'\in \mathcal P_{:i-1}} q^t_{\pi'}}= \frac{\sum_{\pi\in \mathcal P_{:i}}w^t_{\pi}}{\sum_{\pi'\in \mathcal P_{:i-1}} w^t_{\pi'}},
        \end{align*}
        where we exploited the definition of the $q_{\pi}^t$. We apply \Cref{cl:prod} on the right-most term of the previous formula to express the path weights in terms of the edge weights. To simplify notation, for any fixed path $\pi$, we denote with $\pi_{:i}$ its prefix of the first $i$ edges (included). We get:
        \begin{align*}
            \mathbb P ((e_1,e_2, \dots, e_{i}) \in \pi^t &\mid (e_1,e_2, \dots, e_{i-1}) \in \pi^t) =\frac{\sum_{\pi\in \mathcal P_{:i}}w^t_{\pi}}{\sum_{\pi'\in \mathcal P_{:i-1}} w^t_{\pi'}}\\
            &=\frac{\sum_{\pi\in \mathcal P_{:i}} \prod_{e \in \pi} w_e^t}{\sum_{\pi'\in \mathcal P_{:i-1}} \prod_{e' \in \pi'} w_{e'}^t} \\
            &= \frac{ w_{e_i}^t \sum_{p\in \mathcal P_{:i}} \prod_{e \in p_{:i-1}} w_e^t}{w_{e_i}^t \sum_{p\in \mathcal P_{:i}} \prod_{e \in p_{:i-1}} w_e^t + w_{(u,d)}^t \sum_{p\in \mathcal P_{:i}} \prod_{e \in p_{:i-1}} w_e^t} \\
            &= \frac{ w_{e_i}^t  w_u^t}{w_{e_i}^t  w_u^t + w_{(u,d)}^t w_d^t} = q_e^t,
        \end{align*}
        where the second to last equality follows by \Cref{cl:node} and the last one by definition of $q_e^t.$
        \end{proof}

\subsection[Proofs Omitted from Section 6]{Proofs Omitted from \Cref{sec:smooth_lower}}
    \label{app:smooth_lower}
        Before 
    {providing the formal proofs that we omitted from}
    \Cref{sec:smooth_lower}, we recall some results from the literature on the ``one-shot'' version of the problem. We recall the definition of virtual valuation, crucial in characterizing revenue-maximizing auctions \citep{myerson81,Guth86}.

        \begin{definition}(Virtual Valuation)
        Let $X$ be a random variable with probability density function $f_X :[a,b] \to \mathbb{R}_+$ and cumulative density function $F_X$. If $f_X$ is strictly positive on its domain, we can define the virtual valuation function $\phi_X:[a,b] \to \mathbb{R}$ as follows:
        \[
            \phi_X(x) = x - \frac{1-F(x)}{f(x)}.
        \] 
        The random variable $X$ is regular if $\phi_X$ is monotone non-decreasing on its domain. 
    \end{definition}
    \begin{example}
        The virtual valuations of a uniform random variable $U$ in $[a,b]$ is 
        \(
            \phi_U(x) = 2x - b. 
        \)
    \end{example}

    We present a simple lemma that is crucial to the following construction. 
    
        \begin{lemma}
        \label{lem:squares}
            Consider agents whose valuations are drawn from a uniform distribution on a rectangle $R = [a,b] \times [c,d]$. The following properties hold:
            \begin{itemize}
                \item[(i)] For any mechanism $M$ with allocation region $A_M$, it holds that \[
                \E{\rev (M)} = \frac 1{(b-a)(d-c)}\int_{A_M \cap R} (2(v_1 + v_2) - (b+d)) dv_1 dv_2
                \]
                \item[(ii)] The optimal mechanism $M^\star$ allocates if and only if $2(v_1 + v_2) \ge (b+d)$.
                \item[(iii)] If $2(a + c) \ge b+d$, then the optimal mechanism always allocates, with $\E{\rev(M^\star)} = a+c$
            \end{itemize}
        \end{lemma}
        \begin{proof}
            Denote with $V_1$ and $V_2$ the two independent uniform random variables that represent the two valuations. When the valuations of the two agents are drawn by two independent and regular distributions $V_1$ and $V_2$, then the expected revenue is equal to the expected virtual surplus (e.g., Theorem 13.10 of \citet{NisanRTV07}\footnote{Their result is stated for $a = 0$, but can be derived for general $a$ with minimal effort.}), this proves point (i). The optimal mechanism then allocates if and only if the sum of the virtual valuations is positive, i.e., if the realized $v_1$ and $v_2$ are such that
            \begin{equation}
            \label{eq:virtual}
                \phi_{V_1}(v_1) + \phi_{V_2}(v_2) = 2(v_1 + v_2) - (b+d)\ge 0.
            \end{equation}
            Stated differently, the optimal mechanism allocates if and only if $2(v_1 + v_2) \ge (b+d)$. This proves point (ii). Finally, the expected revenue is then easy to compute when $2(v_1 + v_2) \ge 2 (a+c) \ge (b+d)$, i.e., when the trade always happens: with probability $1$, the first agent pays $a$, and the second $c$. Equivalently, This can be computed by integrating the virtual valuation on the allocation region:
            \begin{align*}
                \E{\rev(M^\star)} = \frac{1}{(b-a)(d-c)}\int_{a}^b \int_{c}^d [2(v_1 + v_2)-(b+d)] dv_2 dv_1 = a+c.
            \end{align*}
            {This shows (iii) and completes the proof.}
        \end{proof}

    
    \Domination*
        \begin{proof}
            Let $M$ be any mechanism and $\alpha$ any parameter in $(\nicefrac 4{15}, \nicefrac 35)$. We prove the Lemma via a case analysis, according to the relative position of its allocation region $A_M$ and the support of $\Pb^{\alpha}$.
            \paragraph{Case I} If $Q_1 \cup Q_2 \subseteq A_M$, $M$ always trades the item, and $M_1$ clearly dominates $M$ (by monotonicity, the highest price any valuation in $A_M$ can be charged is $\nicefrac 12$ for each agent, which is what $M_1$ does).
            \paragraph{Case II} If $Q_1 \cap A_M = \emptyset$, then $M$ only considers valuations in $Q_2$, but we know that the optimal mechanism that considers only $Q_2$ is $M_2$ (by \Cref{lem:squares}).\\

            The first two cases rule out all the mechanisms such that the intersection of $A_M$ and $Q_1$ is not proper ($A_M \cap Q_1 = Q_1$ in the first case and $A_M \cap Q_1 = \emptyset$ in the second). Moreover, if $A_M \cap Q_1 \neq \emptyset$, then $Q_2 \subseteq A_M$, by monotonicity of the allocation rule (for any $(x_1,y_1)$ in $Q_1$ and $(x_2,y_2)$ in $Q_2$ it holds that $x_1 \le x_2$ and $y_1 \le y_2$). For the sake of simplicity, we denote with $\ell_i$ the sides of $Q_1$, naming $\ell_1$ the $(\nicefrac 12, \nicefrac 34)$-$(\nicefrac 34, \nicefrac 34)$ segment and proceeding clock-wise.
            
            \paragraph{Case III} The boundary of $A_M$ intersects $\ell_1$ in $(x,\nicefrac 34)$ and $\ell_3$ in $(y, \nicefrac 12)$, with $\nicefrac 12 \le x \le y \le \nicefrac 34.$ {Then, without loss of generality, we can assume that $M$ is completed by the segments $(x, \nicefrac34) -  (x,1)$ and $(y,\nicefrac12) -(1,\nicefrac12)$. So the} 
            behaviour of $M$ with respect to the valuations in $Q_2$ is fixed: accept them and charge $x + \nicefrac 12$. The behavior of $M$ inside $Q_1$ is not clear. However, the virtual valuation is positive on $Q_1$; thus, by \Cref{lem:squares}, the best thing $A_M \cap Q_1$ can be is the whole 
            $[x,\nicefrac 34]\times [\nicefrac 12,\nicefrac 34]$ 
            rectangle, with an expected revenue of $(x+ \tfrac 12)$ times the probability of a trade. {We thus have} 
            \begin{align*}
                \E{\rev{(M)}} &\leq  (x+ \tfrac 12) [(1-\alpha) + \alpha ({3 - 4x})]\\
                &\le \max_{x \in [\nicefrac 12, \nicefrac 34]}(x+ \tfrac 12) [(1-\alpha) + \alpha ({3 - 4x})]= 1 = \E{\rev(M_1)},
            \end{align*}
            {where we used that} 
            the $\max$ is attained in $x = \nicefrac 12$ for all relevant $\alpha$.

            \paragraph{Case IV} The boundary of $A_M$ intersects $\ell_1$ in $(x,\nicefrac 34)$ and $\ell_2$ in $(\nicefrac 34,y)$, with $\nicefrac 12 \le x,y \le \nicefrac 34.$ Reasoning as in Case III, the best possibility for $M$ is that $A_M\cap M_1  = [x,\nicefrac 34] \times [y,\nicefrac 34]$, for an expected revenue of $(x+y)$ times the probability of a trade. {We thus have:} 
            \begin{align*}
                \E{\rev{(M)}} &\leq (x+ y) [(1-\alpha) + \alpha ({3 - 4x})(3-4y)]\\
                &\le (x+ \tfrac 12) [(1-\alpha) + \alpha ({3 - 4x})]\\
                &\le \max_{x \in [\nicefrac 12, \nicefrac 34]}(x+ \tfrac 12) [(1-\alpha) + \alpha ({3 - 4x})]= 1 = \E{\rev(M_1)},
            \end{align*}
            where the 
            {second}
            inequality 
            can be verified {to hold} 
            numerically 
            for any choice of the parameter $\alpha \in (\nicefrac{4}{15}, \nicefrac{2}{5})$, and of $x$ and $y$ in $[\nicefrac 12, \nicefrac 34]$.

            \paragraph{Case IV} The boundary of $A_M$ intersects $\ell_4$ in 
            $(\nicefrac 12,x)$
            and $\ell_3$ in $(y, \nicefrac 12)$, with $\nicefrac 12 \le y,x \le \nicefrac 34.$ This case is easy: by using \Cref{lem:squares} and reasoning as in Case III the best possibility for $A_M$ is to contain the whole $Q_1$ (which is exactly what $M_1$ does).

            \paragraph{Case V} The boundary of $A_M$ intersects $\ell_4$ in 
            $(\nicefrac 12,x)$
            and $\ell_2$ in $(\nicefrac 34,y)$, with $\nicefrac 12 \le y \le x \le \nicefrac 34.$ This case is analogous to Case IV (by symmetry).
        \end{proof}

        \lowersmooth*
        \begin{proof}
            Let $\eps\in (0,\nicefrac 1{15})$ be a small parameter to be set later, and let $\alpha_0= \nicefrac 13$. We consider three probability distributions from the family introduced in the main body, corresponding to parameters $\alpha_0$, $\alpha_0 + \eps$, and $\alpha_0 - \eps$. For the sake of simplicity, let's rename these measures $\Pb^0$, $\Pb^1$, respectively $\Pb^2$. 
            We use the same random variable $V$ to denote the valuations drawn from the different probability distributions. When we change the underlying measure, we are changing its law. We denote with $\mathbb E^i$ the corresponding expectation. Consider now the push forward measures on $[0,1]^2$ (with the Borel $\sigma$-algebra) induced by these tree measures: $\mathbb P_V^0$, $\mathbb P_V^1$ and $\mathbb P_V^2$. A first, crucial result bounds the difference between these three distributions in terms of the KL divergence:
        \begin{claim}
        \label{cl:KLsmooth}
        It holds that $ \max\left\{\kl{\mathbb P_V^1}{\mathbb P_V^0}, \kl{\mathbb P_V^2}{\mathbb P_V^0}\right\} \le 7 \eps^2$.
        \end{claim}
        \begin{proof}[Proof of \Cref{cl:KLsmooth}]
            We produce the calculations for $\kl{\mathbb P_V^1}{\mathbb P_V^0}$, the ones for $\Pb^2_V$ are analogous. We apply the definition of KL divergence and obtain the following chain of inequalities:
            \begin{align*}
                \kl{\mathbb P_V^1}{\mathbb P_V^0} &= (\alpha_0 + \eps)\log \tfrac{\alpha_0 + \eps}{\alpha_0} + (1-\alpha_0 - \eps)\log \tfrac{(1-\alpha_0 - \eps)}{1-\alpha_0}\\
                &\le \log \left( \tfrac{(\alpha_0 + \eps)^2}{\alpha_0} + \tfrac{(1-\alpha_0 - \eps)^2}{1-\alpha_0}\right)
                = \log(1 + \tfrac 92 \eps^2) \le 7 \eps^2,
            \end{align*}
            where we used the concavity of $\log$  and that  $\log x \le \nicefrac{x}{\log 2}$ for all $x > -1.$
        \end{proof}

        When the sequence of agents' valuations is drawn i.i.d.~under $\mathbb P^0$, the two mechanisms analyzed in \Cref{lem:domination} are optimal: $M_1$ (allocating in $[\tfrac12,1] \times [\tfrac12,1]$) and $M_2$ (allocating in $[\nicefrac 34,1]\times [\nicefrac 34,1]$). Simple calculations show that when the sequence is generated according to $\mathbb P^i$  for $i \in \{1,2\}$, then the optimal mechanism is $M_i.$ Moreover, by \Cref{lem:domination}, it is always convenient for the learner to either post $M_1$ or $M_2$  (as either $M_1$ or $M_2$ dominates $M$ {\em regardless} of the actual distribution).

        Consider now an adversary that selects uniformly at random between $\mathbb P^1$ and $\mathbb P^2$ and extracts the valuations i.i.d.~from it. We prove that any (deterministic)\footnote{It is enough to consider deterministic adversaries thanks to Yao's principle.} learning algorithm $\mathcal A$ suffers at least $\Omega(\sqrt{T})$ regret against this adversary, meaning that it suffers the same regret against at least one of the two distributions, thus proving the Theorem. In particular, note that every time the learner chooses the suboptimal mechanism, then it suffers instantaneous regret $\Omega(\eps)$:
        \begin{align}
        \label{eq:suboptimal}
            \mathbb{E}^1[\rev(M_1) - \rev(M_2)] = \mathbb{E}^2[\rev(M_2) - \rev(M_1)] = \tfrac 32 \eps.
        \end{align}

        Thus, to show a lower bound on the regret of the learner, it suffices to show an upper bound on how often it chooses the right mechanism.  
        Denote with $N_1^T$ and $N_2^T$ the random variables counting the number of times that $\mathcal A$ plays mechanism $M_1$, respectively $M_2$ (there is no point in posting any other mechanism). 
        For each $j \in \{0,1,2\}$, consider the run of $\mathcal A$ against the stochastic adversary which draws $V^1,V^2, \dots $ i.i.d.~from $\mathbb{P}^j$. For $i=1,2$, we have the following: 
        \begin{align}
        \nonumber
            \Ei{N_i^T} - \Eo{N_i^T} &= \sum_{t=2}^T \Pi{M^t = M_i} - \Po{M^t = M_i}\\
            &\le\sum_{t=2}^T ||\mathbb P^i_{V^1, \dots V^{t-1}} - \mathbb P^0_{V^1, \dots V^{t-1}}||_{\tv} \tag{\text{Definition of total variation}}\\
            &\le\sum_{t=2}^T \sqrt{\frac t2 \kl{\mathbb P^i_V}{\mathbb P^0_V}} \tag{By Pinsker's inequality on i.i.d. r.v.}\\
            \label{eq:KL}
            &\le 2\eps T \sqrt{T},
        \end{align}
        where in the last step we applied \Cref{cl:KLsmooth}. Note, $\mathbb P^j_{V^1, \dots V^t}$ is the push-forward measure on $([0,1]^2)^t$ induced by $t$ i.i.d. draws of $V$ from distribution $\mathbb{P}^j$, $j \in \{0,1,2\}$.  
        Averaging \Cref{eq:KL}, we get
        \begin{equation}
            \label{eq:last_step}
            \frac 12 \sum_{i=1,2}\Ei{N_i^T} \le \frac 12 \sum_{i=1,2}\Eo{N_i^T} + 2\eps T \sqrt{T} = \left(\frac 12 + 2 \eps \sqrt{T}\right) T.
        \end{equation}         
        We have all the ingredients to conclude the proof: the regret of the learning algorithm $\mathcal{A}$ is at least
        \begin{align}
            \nonumber
            R_T(\mathcal A) &\ge \frac 12 \sum_{i=1,2}\Ei{\sum_{t = 1}^{T} \revt(M_i) - \sum_{t = 1}^{T} \revt(M^t)}\\
            &\ge \frac 3 4 \eps \sum_{i=1,2}(T-\Ei{N_i^T}) \ge \frac 3 4 \eps \left(1 - 4 \eps \sqrt{T}\right)\ge \frac{3}{64} \sqrt{T}  \tag*{\text{(By \Cref{eq:last_step,eq:suboptimal})}}
        \end{align}
        where the last inequality follows by setting $\eps = \tfrac{1}{16\sqrt{T}}$).
        \end{proof}


\end{document}